\documentclass[journal]{IEEEtran}
\ifCLASSINFOpdf
\else
\fi
\ifCLASSOPTIONcompsoc
  \usepackage[caption=false,font=normalsize,labelfont=sf,textfont=sf]{subfig}
\else
  \usepackage[caption=false,font=footnotesize]{subfig}
\fi

\usepackage{amsthm}

\usepackage{amssymb}
\usepackage{mathtools}
\usepackage{cases}
\usepackage{autobreak}

\usepackage{booktabs}
\usepackage{multirow}

\usepackage{color}
\usepackage{hyperref}
\usepackage{cite}

\usepackage{amsmath}
\usepackage{extarrows}
\usepackage{lipsum}
\usepackage{caption}

\usepackage{comment}

\usepackage{algorithmic}
\usepackage{algorithm}
\usepackage{bbm}
\usepackage{bm}

\newcommand{\E}{\mathbb{E}}
\newcommand{\Prob}{\mathbb{P}}

\newcommand{\cov}{\mathrm{Cov}}
\newcommand{\Var}{\mathbb{V}}

\newcommand{\LR}{\bold{R}^{\frac{1}{2}}}

\newcommand{\RT}{\bold{T}^{\frac{1}{2}}}
\newcommand{\MS}{\bold{S}^{\frac{1}{2}}}

\newcommand{\FR}{{\bold{R}}}
\newcommand{\FT}{{\bold{T}}}
\newcommand{\BK}{{\bold{K}}}

\newcommand{\FS}{{\bold{S}}}
\newcommand{\BQ}{{\bold{Q}}}
\newcommand{\BH}{{\bold{H}}}

\newcommand{\BI}{{\bold{I}}}
\newcommand{\BX}{{\bold{X}}}

\newcommand{\BZ}{{\bold{Z}}}
\newcommand{\BY}{{\bold{Y}}}

\newcommand{\BG}{{\bold{G}}}

\newcommand{\BT}{{\bold{T}}}
\newcommand{\Bh}{{\bold{h}}}

\newcommand{\Bx}{{\bold{x}}}

\newcommand{\By}{{\bold{y}}}
\newcommand{\Bw}{{\bold{w}}}
\newcommand{\Bz}{{\bold{w}}}
\newcommand{\BA}{{\bold{A}}}
\newcommand{\BB}{{\bold{B}}}
\newcommand{\BC}{{\bold{C}}}

\newcommand{\BW}{{\bold{W}}}

\newcommand{\BF}{{\bold{F}}}

\newcommand{\BO}{{\mathcal{O}}}

\DeclareMathOperator{\Tr}{Tr}

\newcommand{\RNum}[1]{\uppercase\expandafter{\romannumeral #1\relax}}

\newtheorem{remark}{Remark}
\newtheorem{theorem}{Theorem}
\newtheorem{lemma}{Lemma}
\newtheorem{proposition}{Proposition}
\newtheorem{subremark}{Remark}[remark]

\newcommand\numberthis{\addtocounter{equation}{1}\tag{\theequation}}

\allowdisplaybreaks[4]
\begin{document}
%
\title{Fundamental Limits of Optical Fiber MIMO Channels With Finite Blocklength}
%
%

\author{Xin~Zhang,~Dongfang~Xu,~Shenghui~Song,~\IEEEmembership{Senior Member,~IEEE}, and M\'erouane Debbah, \IEEEmembership{Fellow, IEEE}
\thanks{X. Zhang, D. Xu, and S. H. Song are with the Department of Electronic and Computer Engineering, The Hong Kong University of Science and Technology, Hong Kong (E-mail: xzhangfe@connect.ust.hk; \{eedxu,  eeshsong\}@ust.hk).

M. Debbah is with Khalifa University of Science and Technology, P O Box
127788, Abu Dhabi, UAE. (E-mail: merouane.debbah@ku.ac.ae)
}}
\maketitle

\begin{abstract} The multiple-input and multiple-output (MIMO) technique is regarded as a promising approach to boost the throughput and reliability of optical fiber communications. However, the fundamental limits of optical fiber MIMO systems with finite block-length (FBL) are not available in the literature. This paper studies the fundamental limits of optical fiber multicore/multimode systems in the FBL regime when the coding rate is a perturbation within $\BO(\frac{1}{\sqrt{ML}})$ of the capacity, where $M$ and $L$ represent the number of transmit channels and blocklength, respectively. Considering the Jacobi MIMO channel, which was proposed to model the nearly lossless propagation and the crosstalks in optical fiber systems, we derive the upper and lower bounds for the optimal error probability. For that purpose, we first set up the central limit theorem for the information density in the asymptotic regime where the number of transmit, receive, available channels and the blocklength go to infinity at the same pace. The result is then utilized to derive the upper and lower bounds for the optimal average error probability with the concerned rate. The derived theoretical results reveal interesting physical insights for Jacobi MIMO channels with FBL. First, the derived bounds for Jacobi channels degenerate to those for Rayleigh channels when the number of available channels approaches infinity. Second, the high signal-to-noise (SNR) approximation indicates that a larger number of available channels results in a larger error probability. Numerical results validate the accuracy of the theoretical results and show that the derived bounds are closer to the performance of practical LDPC codes than outage probability.
\end{abstract}


%
\IEEEpeerreviewmaketitle

\section{Introduction}

Due to its capability in establishing long-distance communications with low-level loss, optical fiber communications have played a crucial role in telecommunication systems~\cite{keiser2000optical,singer2008electronic}. To meet the increasing demand for high data rates, many innovative techniques, including wavelength-division multiplexing (WDM) and polarization-division multiplexing (PDM), have been proposed to fully exploit the degree of freedom of optical fiber communications~\cite{winzer2010beyond,winzer2011mimo}. Among them, space division multiplexing (SDM) is considered as a promising approach for the next generation optical fiber systems~\cite{marom2015switching,tarighat2007fundamentals}, due to its ability to create multiple parallel transmission channels (paths or modes) within the same fiber and potentially improve the throughput multiple times.

The multiple channels of optical fiber communications correspond to multiple modes or multiple cores in the fiber, which enable the multiple-input multiple-output (MIMO) technique to boost the data rate and reliability. However, there exists a challenging issue for multiple parallel transmissions, namely, the crosstalk between different modes, which is caused by imperfections, twists, and the bending of the fiber~\cite{fini2010statistics} and inevitably introduces the coupling between channels. Furthermore, according to Winzer~\textit{et al.}'s analysis~\cite{winzer2011mimo}, the scattering matrix is a unitary matrix, due to the near-lossless propagation. Some efforts have been devoted to characterize the signal propagation in optical fiber MIMO channels but failed to characterize the unitary attribute of signal propagation~\cite{shah2005coherent,hsu2006capacity}. In fact, the crosstalk between different modes and the unitary scattering matrix make optical fiber MIMO channels different from wireless MIMO channels and require new efforts to fully exploit their fundamental limits.


To this end, the Jacobi model, which truncates a random unitary matrix, was proposed to characterize the optical fiber MIMO channel with strong crosstalk and weak backscattering, and has attracted many research interests~\cite{dar2012jacobi}. Specifically, Dar~\textit{et al.}~\cite{dar2012jacobi} gave the closed-form evaluation for the ergodic capacity and outage probability of Jacobi MIMO channels. Karadimitrakis~\textit{et al.}~\cite{karadimitrakis2014outage}  derived the closed-form expression for the outage probability in the asymptotic regime where the number of transmit, receive, and available channels go to infinity at the same pace and parameterized the transmission loss with the number of unaddressed channels. Nafkha~\textit{et al.}~\cite{nafkha2017upper} gave the closed-form upper and lower bounds for the ergodic capacity. Nafkha~\textit{et al.}~\cite{nafkha2020closed} derived a new closed-form expression for the ergodic capacity and gave the evaluation for the ergodic sum capacity of Jacobi MIMO channels with the minimum mean squared error receiver. Wei~\textit{et al.}~\cite{wei2018exact} derived the explicit expressions for the exact
moments of mutual information (MI) in the high signal-to-noise ratio (SNR) regime and the approximation for the outage probability. Laha~\textit{et al.}~\cite{laha2021optical} derived the ergodic capacity with arbitrary transmit covariance over Jacobi MIMO channels.  


The above works investigated the capacity and outage probability of Jacobi MIMO channels in the infinite blocklength (IBL) regime. However, many innovative applications, such as autonomous driving, virtual/augmented reality, and industrial automation for real-time internet of things (IoT)~\cite{she2017radio}, pose stringent latency requirements on communication systems. To this end, ultra-reliable and low-latency communications (URLLC) with short-length codes must be considered. Unfortunately, existing IBL analyses fail to characterize the impact of the blocklength, and the finite blocklength (FBL) analysis for optical fiber MIMO systems requires in-depth investigation.

In the FBL regime~\cite{polyanskiy2010channel}, the conventional Shannon's coding rate was refined to show that the maximal channel coding rate can be represented by
\begin{equation}
\log M(L,\varepsilon)=LC-\sqrt{LV}Q^{-1}(\varepsilon)+\BO(\log(L)),
\end{equation}
where $C$ denotes the channel capacity, $V$ represents the channel dispersion, and $Q^{-1}(\cdot)$ is the inverse $Q$-function. Here $M(L,\varepsilon)$ represents the cardinality of a codebook with blocklength $L$, which can be decoded with error probability less or equal to $\varepsilon$. The characterization of the trade-off between error probability, blocklength, and coding rate is challenging. For the FBL analysis of single-input single-output systems, Polyanskiy~\textit{et al.}~\cite{polyanskiy2010channel} and Hayashi~\cite{hayashi2009information} derived the second-order coding rate for the additive white Gaussian noise (AWGN) channel. Polyanskiy~\textit{et al.}~\cite{polyanskiy2011scalar} evaluated the dispersion for coherent scalar fading channels and Zhou~\textit{et al.}~\cite{zhou2019lossy} investigated the FBL performance in terms of second-order asymptotics over parallel AWGN channels with quasi-static fading. Yang~\textit{et al.}~\cite{yang2014quasi} investigated the maximal achievable rate for quasi-static multiple-antenna systems with a given blocklength and error probability, and Collins~\textit{et al.}~\cite{collins2018coherent} derived the second-order coding rate for MIMO block fading channels. Hoydis~\textit{et al.}~\cite{hoydis2015second} and Zhang~\textit{et al.}~\cite{zhang2022second} studied the optimal average error probability for quasi-static Rayleigh and Rayleigh-product MIMO channels, respectively, when the rate is close to capacity, and the results showed that outage probability is optimistic in characterizing the error probability. To the best of the authors' knowledge, there is only one work considering the impact of codelength on Jacobi MIMO systems~\cite{opkaradimitrakis2017gallager}, in which the error exponent was derived when the number of receive channels is larger than that of transmit channels. However, the fundamental limits of optical fiber MIMO channels with FBL are not available in the literature, which will be the focus of this work.

\textit{Challenges:} The characterization of the optimal average error probability for Jacobi MIMO channels with FBL resorts to the distribution of information density (ID), which is challenging to obtain due to two reasons. First, the channel model is complex. Jacobi MIMO channels are modeled by the Beta matrix, which is the product of a Wishart matrix and an inverse Wishart matrix. As a result, the inverse structure and product of random matrices in the Jacobi model, together with the fluctuations of two random matrices, must be handled. In particular, setting up the CLT for ID can be achieved by showing that the characteristic function of ID converges to that of a Gaussian distribution, but the complex channel structure makes the computation much involved. Second, ID consists of not only the MI term but also two additional terms~\cite{hoydis2015second}. Thus, the characterization of ID is more complex than that for MI~\cite{karadimitrakis2014outage}, because the asymptotic variance of the additional terms and their asymptotic covariance with MI must be evaluated.

\textit{Contributions}: In this paper, we investigate the optimal average error probability of optical fiber MIMO systems with FBL. The contributions of this work are summarized as follows:


\begin{itemize}

\item[(i)] The closed-form approximation for the ergodic capacity of Jacobi MIMO channels is derived in the asymptotic regime where the number of transmit, receive, and available channels go to infinity at the same pace. With this result, a CLT for ID is set up when the blocklengh approaches infinity with the same pace as the number of channels, which proves the asymptotic Gaussianity of ID with closed-form mean and variance. Besides, the approximation error of the cumulative distribution function (CDF) is shown to be $\BO(L^{-\frac{1}{4}})$, where $L$ denotes the blocklength. The result can degenerate to that for Rayleigh channels in~\cite[Theorem 2]{hoydis2015second} when the number of available channels has a higher order than the blocklength, the number of transmit and receive channels.


\item[(ii)] Based on the CLT, closed-form expressions for the upper and lower bounds of the optimal average error probability are derived and the bounds can degenerate to existing results. Specifically, when the blocklength approaches infinity with a higher order than the number of channels, both the upper and lower bounds approach the outage probability~\cite{karadimitrakis2014outage}. Furthermore, the bounds for Jacobi channels converge to those for Rayleigh channels~\cite{hoydis2015second} when the number of available channels increases to infinity with a higher order than the blocklength, the number of transmit and receive channels. When the rate is close to the capacity, the dispersion in the upper bound agrees with the error exponent for Rayleigh channels~\cite{karadimitrakis2017gallager} if the number of available channels has a higher order. To evaluate the impact of the number of available channels, high SNR approximations for the bounds are derived to show that a larger number of available channels will result in a larger error probability.


\item[(iii)] Simulation results validate the accuracy of the derived bound. It is shown that the gap between the upper and lower bounds for the optimal average error probability is small in the practical SNR regime and the derived bounds are closer to the error probability of practical LDPC codes than outage probability. In fact, the gap between the bounds and outage probability is not ignorable for small blocklength. This indicates that the derived bounds provide a better performance analysis when the rate is close to the capacity. 


\end{itemize}




\textit{Paper Outline:} The rest of this paper is organized as follows. Section~\ref{sec_mod} introduces the system model and problem formulation. Section~\ref{sec_main} gives the CLT for ID and the upper and lower bounds for the optimal average error probability over Jacobi MIMO channel. Section~\ref{sec_simu} validates the theoretical results by the numerical simulations and Section~\ref{sec_con} concludes the paper. The notations in this paper are defined as follows.

\textit{Notations:}  The vector and matrix are denoted by the bold, lower case letters and bold, upper case letters, respectively. The $(i,j)$-th entry of $\bold{A}$ is denoted by $[A]_{i,j}$ or $A_{i,j}$. The conjugate transpose, trace, and spectral norm of $\bold{A}$ are represented by $\bold{A}^{H}$, $\Tr(\BA)$, and $\|\BA \|$, respectively. The $N$-by-$N$ identity matrix is denoted by $\BI_{N}$. The space of $N$-dimensional complex vectors and $M$-by-$N$ complex matrices are represented by $\mathbb{C}^{N}$ and $\mathbb{C}^{M\times N}$, respectively. The expectation of $x$ is denoted by $\E [x]$ and the centered $x$ is represented by $\underline{x}=x-\E [x]$. The covariance of $x$ and $y$ is denoted by $\cov(x,y)=\E\underline{x}\underline{y} $ and the CDF of the standard Gaussian distribution is given by $\Phi(\cdot)$. The conjugate of $x$ is denoted by $(x)^{*}$. The limit that $a$ approaches $b$ from the right is represented by $a  \downarrow b$ and the support operator is represented by $\mathrm{supp}(\cdot)$. The probability operator is denoted by $\Prob(\cdot)$ and the probability measure whose support is a subset of $\mathcal{S}$ is denoted by $\mathrm{P}(\mathcal{S})$. The almost sure convergence and convergence in probability are represented by $\xrightarrow[]{a.s.}$ and $\xrightarrow[]{\mathcal{D}}$, respectively. The big-O and little-o notations are represented by $\BO(\cdot)$ and $o(\cdot)$, respectively. 





\section{System Model and Problem Formulation}
\label{sec_mod}
\subsection{System Model}
As shown in Fig.~\ref{optical_model}, we consider a single-segment optical fiber system with total $n$ available channels~\cite{karadimitrakis2014outage}, where there are $M \le n$ excited transmit channels and $N \le n$ excited receive channels, and the propagation in the fiber is near-lossless. This corresponds to the communication system utilizing multicore/multimode fibers~\cite{ryf2018high}. A strong crosstalk (shown as arrows between modes in Fig.~\ref{optical_model}) between channels (modes) is considered and the backscattering is neglected. The propagation in the concerned system can be characterized by the scattering matrix $\BK\in \mathbb{C}^{2n\times 2n}$ with~\cite{dar2012jacobi,karadimitrakis2014outage,winzer2011mimo}
\begin{equation}
\BK=
\begin{bmatrix}
\bold{R}_{\mathrm{l}} & \bold{T}_{\mathrm{l}}
\\
\bold{T}_{\mathrm{r}} & \bold{R}_{\mathrm{r}}
\end{bmatrix},
\end{equation}
which connects the $n$ left modes with the $n$ right modes. Here, $\FR_{\mathrm{l}}$ and $\FR_{\mathrm{r}}$ denote of the left-to-left and right-to-right reflection coefficients, respectively. Specifically, $[\FR_{\mathrm{l}}]_i$ is the output of the left $n$ modes induced by inserting a unit-amplitude signal into the $i$-th mode at the left hand side. The definition of $\FR_{\mathrm{r}}$ is similar. $\FT_{\mathrm{l}}$ and $\FT_{\mathrm{r}}$ represent the left-to-right and right-to-left transmission coefficients, respectively, where $[\FT_{\mathrm{r}}]_i$ is the output at the right hand side induced by inserting a unit-amplitude signal into the $i$-th mode at the left hand side. 

Given the lossless assumption, the input power and output power induced by the input $\bold{v}_{\mathrm{input}}$ are equal, such that
\begin{equation}
\bold{v}_{\mathrm{input}}^{H}\bold{v}_{\mathrm{input}}=\bold{v}_{\mathrm{output}}^{H}\bold{v}_{\mathrm{output}}=\bold{v}_{\mathrm{input}}^{H}\BK^{H}\BK\bold{v}_{\mathrm{input}},
\end{equation}
where $\bold{v}_{\mathrm{output}}=\BK\bold{v}_{\mathrm{input}}$ denotes the output. Thus, the scattering matrix $\BK$ is unitary. Given the negligibility of the backscattering effect in optical fiber and the symmetric attribute of both sides, there holds true that $\FR_{\mathrm{l}}=\FR_{\mathrm{r}}\approx \bold{0}_n$, $\FR_{\mathrm{l}}=\FR_{\mathrm{l}}^{T}$, $\FR_{\mathrm{r}}=\FR_{\mathrm{r}}^{T}$, and $\FT_{\mathrm{l}}=\FT_{\mathrm{r}}^{T}$~\cite{karadimitrakis2014outage}. Thus, the eigenvalues of the four matrices $\FT_{\mathrm{l}}\FT_{\mathrm{l}}^{H}$, $\FT_{\mathrm{r}}\FT_{\mathrm{r}}^{H}$, $\BI_{n}-\FR_{\mathrm{l}}\FR_{\mathrm{l}}^{H}$, and $\BI_{n}-\FR_{\mathrm{r}}\FR_{\mathrm{r}}^{H}$ are all from the same set consisting of $n$ elements, $0\le \lambda_i \le 1$, $i=1,2,...,n$. The transmission matrix $\FT_l$ and $\FT_r$ are modeled as $n\times n$ Haar-distributed unitary matrices~\cite{dar2012jacobi}.

Given the excited channels at the transmit and receive side are of size $M$ and $N$, respectively, the effective channel matrix, denoted by $\BH\in \mathbb{C}^{N\times M}$, is a truncated version of $\FT_{\mathrm{l}}\in\mathbb{C}^{n\times n} $. In this case, the received signal at the $t$-th slot, $\bold{r}_t\in \mathbb{C}^{N}$, can be represented by
\begin{equation}
\label{sig_mod}
\bold{r}_{t}={\BH}\bold{s}_{t}+\sigma\bold{w}_{t}, ~~t=1,2,...,L,
\end{equation}
where $\bold{s}_{t}\in\mathbb{C}^{M}$ represents the transmit signal (channel input) and $ \bold{w}_{t} \in\mathbb{C}^{N}$ denotes the AWGN, whose entries follow $\mathcal{CN}(0,1)$. Here $\sigma^2$ and $L$ represent the noise power and the blocklength, respectively. In the following, we introduce the Jacobi model for the effective channel matrix ${\BH}\in\mathbb{C}^{N\times M}$.

\subsection{Jacobi Model}
We first define the Jacobi ensemble, which will be utilized to characterize the singular value distribution for the channel matrix $\BH$. The Jacobi ensemble can be represented by 
\begin{equation}
\bold{J}(p,q,r)=\BX\BX^{H}\left(\BX\BX^{H}+\BY\BY^{H}  \right)^{-1},
\end{equation}
where $p\le q $, $p \le r$, and $\BX \in \mathbb{C}^{p \times q}$, $\BY \in \mathbb{C}^{p \times r}$ are two independent and identically distributed (i.i.d.) complex Gaussian random matrices with variance $\frac{1}{p}$. The channel matrix ${\BH}\in\mathbb{C}^{N\times M}$ can be characterized by the Jacobi random matrix~\cite{dar2012jacobi}. Specifically, when $N\le M$ and $N+M\le n$, $\BH\BH^{H}$ has the same eigenvalue distribution as $\bold{J}(N,M,n-M)$. When $ M<N$ and $N+M\le n$, $\BH^{H}\BH$ has the same eigenvalue distribution as $\bold{J}(M,N,n-N)$. When $N+M>n$, $M+N-n$ eigenvalues of $\BH\BH^{H}$ are $1$ and the distribution of the rest $n-\max\{M,N\}$ eigenvalues are same as that of $\bold{J}(n-\min\{M,N\},\min\{M,N\},n-\max\{M,N\})$. Thus, without loss of generality, we only investigate the case $M+N \le n$. Furthermore, we consider the quasi-statistic case where $\BH$ does not change in $L$ channel uses, and assume that the perfect channel state information (CSI) is available at the receiver. For ease of illustration, we introduce the following notations:
$\FS^{(L)}=(\bold{s}_1,\bold{s}_2,...,\bold{s}_L)$, $\FR^{(L)}=(\bold{r}_1,\bold{r}_2,...,\bold{r}_L)$, and $\BW^{(L)}=(\bold{w}_1,\bold{w}_2,...,\bold{w}_L)$. With above settings, we will investigate the optimal average error probability over Jacobi MIMO channels with blocklength $L$. In the following, we first define the performance metrics.


\begin{figure}[t!]
\centering\includegraphics[width=0.45\textwidth]{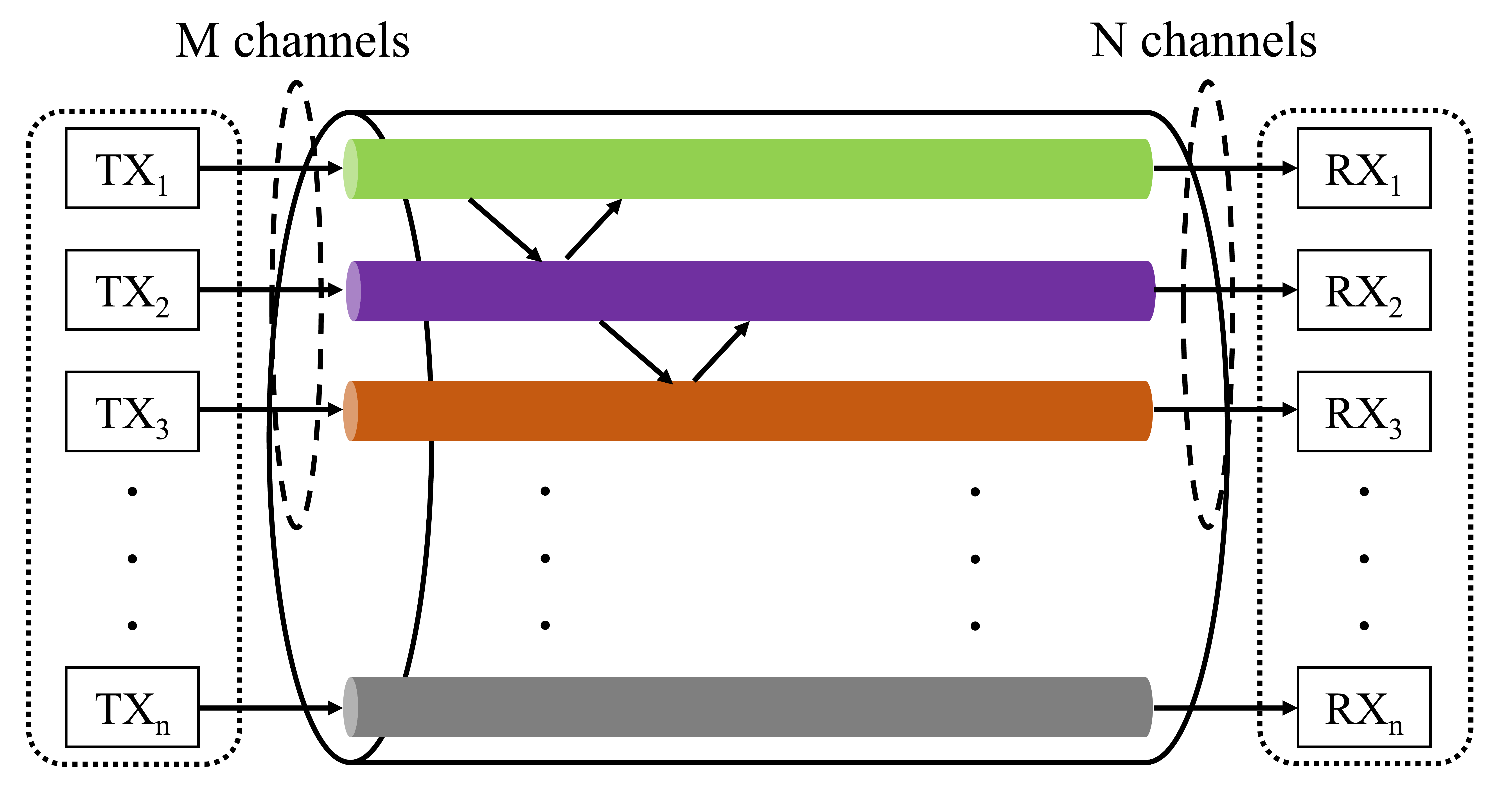}
\caption{Optical SDM MIMO systems with $n$ channels.}
\label{optical_model}
\vspace{-0.3cm}
\end{figure}

\subsection{Optimal Average Error Probability}
A code for the above system can be represented by the following encoding and decoding mapping.

\textit{Encoding mapping} generates the coded message by the mapping from the message $m \in \mathcal{M}$ to the code $\FS_{m}^{(L)} \in \mathbb{C}^{M \times L}$, and can be represented by 
\begin{equation}
\varrho :\mathcal{M} \rightarrow \mathbb{C}^{M\times L}.
\end{equation}
Thus, the transmitted symbol is given by $\FS_{m}^{(L)}=\varrho(m)$, where $m$ is uniformly distributed in $\mathcal{M}=\{1,2,..,G \}$ and $\mathcal{C}_{L}$ denotes the codebook, i.e., $\mathcal{C}_{L}=\{\varrho(1),\varrho(2),...,\varrho(G)\}$. Given the limited transmit power, we consider the maximal energy constraint which requires $\mathrm{supp}(\mathcal{C}_L)\subseteq \mathcal{S}^{L}$, where 
\begin{equation}
\label{max_cons}
\begin{aligned}
\mathcal{S}^{L}=\Bigg\{\FS^{(L)} \in\mathbb{C}^{M\times L }| \frac{\Tr(\FS^{(L)}(\FS^{(L)})^{H})}{ML} \le 1  \Bigg\}.
\end{aligned}
\end{equation}

\textit{Decoding mapping} recovers the message from the channel output $\FR^{(L)}=\BH \varrho(m)+\sigma \BW^{(L)}$, and can be represented by
\begin{equation}
\varpi :\mathbb{C}^{N\times L} \rightarrow \mathcal{M} \cup \{\mathrm{e}\}.
\end{equation}
The mapping $\varpi$ makes the decision $\hat{m}=\varpi(\FR^{(L)})$. If $\hat{m}=m$, the message is correctly decoded. Otherwise, an error $\mathrm{e}$ happens.

Given the message $m$ is uniformly distributed, the \textit{average error probability} for a codebook $\mathcal{C}_L$ of the message set $\mathcal{M}$ with blocklength $L$ is given by
\begin{equation}
\label{pe_def}
\mathrm{P}_{\mathrm{e}}^{(L)}(\mathcal{C}_L)=\frac{1}{G}\sum_{i=1}^{G}\Prob( \hat{m}\neq m | m=i),
\end{equation}
where the evaluation involves the randomness of $\BH$, $\BW^{(n)}$, and $\FS^{L} \in \mathcal{S}^{L}$. The optimal average error probability is given by
\begin{equation}
\label{pro_e_ori}
\begin{aligned} 
\mathrm{P}_{\mathrm{e}}^{(L)}(R)=\inf_{\mathrm{supp}(\mathcal{C}_L)\subseteq \mathcal{S}^{L}}
\mathrm{ P}_{\mathrm{e}}^{(L)}(\mathcal{C}_L) ,
\end{aligned}
\end{equation}
where $R$ denotes the per-antenna rate of each transmitted symbol and $\frac{1}{ML}\log(|\mathcal{C}_L|) \ge R $. Our goal is to obtain the bounds for the optimal average error probability with the maximal energy constraint in~(\ref{max_cons}). The optimal average error probability can be characterized by the distribution of ID, which is introduced in the following.

\subsection{Information Density (ID) and Error Bounds}
As shown in~\cite{polyanskiy2010channel,collins2018coherent,hoydis2015second}, the optimal average error probability can be bounded by the CDF of ID. The ID of the considered MIMO systems is given by
\begin{align*}
\label{mid_exp}
& I_{N,M,n,L}^{\BW,\BH} (\sigma^2)\overset{\bigtriangleup}{=}\frac{1}{M}\log\det(\bold{I}_{N}+\frac{1}{\sigma^2}\BH\BH^{H}) 
+\frac{1}{ML} \times
\\
&
\Tr((\BH\BH^{H}+\sigma^2\bold{I}_{N})^{-1}
(\BH\FS^{(L)}\!+\!\sigma \BW^{(L)})(\BH\FS^{(L)}\!+\!\sigma \BW^{(L)})^{H} )
\\
&
-\frac{1}{ML}\Tr(\BW^{(L)}(\BW^{(L)})^{H}).\numberthis
\end{align*}
It can be observed from~(\ref{mid_exp}) that the first term of ID is the per-antenna MI (per-antenna capacity), whose distribution can be utilized to determine the outage probability, which serves as the lower bound for the optimal average error probability in the IBL regime. Besides the MI term, ID has two extra terms including a trace of the resolvent for $\BH\BH^{H}$ and a noise related term. As a result, the fluctuation of ID is more complex than that of MI due to the covariance between three terms. To proceed, we show that the optimal error probability can be bounded by the ID distribution. Given the ratios $y_1=\frac{N}{n}$, $y_2=\frac{M}{n}$, and $\beta=\frac{L}{M}$, $M  \xrightarrow[]{(y_1,y_2,\beta)}\infty$ represents the asymptotic regime where $N$, $M$, $n$, and $L$ grow to infinity with the fixed ratios $y_1$, $y_2$, and $\beta$.


In the following, we will consider the rate within $\BO(\frac{1}{\sqrt{ML}})$ of $C(\sigma^2)$, i.e.~\cite{polyanskiy2010channel,yang2013quasi,zhou2018dispersion},
\begin{align}
\label{se_code_rate}
\liminf_{ L  \xrightarrow[]{(y_1,y_2, \beta)}\infty}\frac{1}{\sqrt{M L}} \{\log(|\mathcal{C}_L |)- ML \E [C(\sigma^2)]\}  \ge r,
\end{align}
where $r$ represents the second-order coding rate. Given $r$, the optimal average error probability is given by~\cite{hayashi2009information,hoydis2015second}
\begin{align}
   \label{def_oaep}
& \Prob_{\mathrm{e}}(r| y_1,y_2, \beta )
{=}\liminf\limits_{\mathrm{supp}(\mathcal{C}_L)\subseteq \mathcal{S}^{(L)} }\limsup_{L  \xrightarrow[]{(y_1,y_2, \beta)}\infty} \mathrm{P}_{\mathrm{e}}^{(L)}(\mathcal{C}_L).
\end{align}
According to~\cite[Eq. (77) and Eq. (89)]{hoydis2015second}, the optimal average error probability can be bounded by
\begin{equation}
B(r)\le \Prob_{\mathrm{e}}(r|y_1,y_2, \beta  ) \le U(r),
\end{equation}
where the upper and lower bounds can be respectively given by 
\begin{subequations}
\begin{align}
U(r)&=\lim_{\zeta  \downarrow 0} \limsup\limits_{ L \xrightarrow[]{(y_1,y_2, \beta)} \infty}\nonumber
\\
&
 \Prob\Bigl[\sqrt{ML}(I^{\BW,\BH}_{N,M,n,L}(\sigma^2)-\E [C(\sigma^2)]) \!\le \! r+\zeta  \Bigr]\label{upp_bound},
\\
B(r)&=\inf\limits_{\Prob({\FS^{(L+1)})}\in \mathrm{P}(\mathcal{S}_{=}^{L+1})}
\lim\limits_{\zeta  \downarrow 0} 
 \limsup\limits_{L \xrightarrow[]{(y_1,y_2,\beta)} \infty} \nonumber
 \\
 &
 \Prob\!\Bigl[\sqrt{ML}(I^{\BW,\BH}_{N,M,n,L+1}(\sigma^2)\!-\!\E[C(\sigma^2)])\!\le\! r-\zeta  \Bigr]\label{low_bound}.
\end{align}
\end{subequations}
The upper bound in~(\ref{upp_bound}) is achieved by the spherical Gaussian codebook $\FS_{\mathrm{G}}^{(L)}\in \mathbb{C}^{M\times L}=\widetilde{\bold{G}}^{(L)}\left(\frac{1}{ML}\Tr (\widetilde{\bold{G}}^{(L)}(\widetilde{\bold{G}}^{(L)})^{H})\right)^{-\frac{1}{2}}$, where $\widetilde{\bold{G}}^{(L)}\in \mathbb{C}^{M\times L}$ is an i.i.d. Gaussian matrix~\cite{hoydis2015second}. The lower bound in~(\ref{low_bound}) is obtained by taking $\inf$ operation over the set of probability measures $\{\Prob({\FS^{(L+1)})}\in \mathrm{P}(\mathcal{S}_{=}^{L+1})\}$, where $\mathcal{S}_{=}^{L}$ denotes the equal power constraint given by
\begin{equation}
\label{sph_cons}
\begin{aligned}
\mathcal{S}^{L}_{=}=\Bigg\{\FS^{(L)} \in\mathbb{C}^{M\times L }| \frac{\Tr(\FS^{(L)}(\FS^{(L)})^{H})}{ML} = 1  \Bigg\}.
\end{aligned}
\end{equation}
It is worth mentioning that the transmitted symbols satisfying the equal energy constraint in~(\ref{sph_cons}) also follow the maximal energy constraint in~(\ref{max_cons}). The adaptation from the maximal energy constraint to the equal energy constraint can be obtained by introducing an auxiliary symbol~\cite[Lemma 39]{polyanskiy2010channel}. Moreover, there holds true that $\FS_{\mathrm{G}}^{(L)} \in \mathcal{S}^{L}_{=}$.

\subsection{Problem Formulation}
The evaluation for the upper and lower bounds resorts to the ID distribution with the equal energy constraint~(\ref{sph_cons}) and the CDF of ID can be represented by
\begin{equation}
D(x)=\Prob\Bigl\{ \sqrt{ML}(I^{\BW,\BH}_{N,M,n,L}(\sigma^2)-\E[{C}(\sigma^2)]) \le x  \Bigr\}.
\end{equation}
Unfortunately, it is very difficult to obtain the exact expression of the optimal average error probability for arbitrary $N$, $M$, $n$, and $L$ due to the complex structure of Jacobi MIMO channels. To obtain the closed-form evaluation for $D(x)$, we will adopt the asymptotic regime $M  \xrightarrow[]{(y_1,y_2,\beta)}\infty$ and investigate the distribution of ID. 


\section{Main Results}
\label{sec_main}
In this section, we first introduce the asymptotic regime and give the closed-form approximation for the ergodic capacity. Then, we set up a CLT for ID in~(\ref{mid_exp}) with closed-form mean and variance, and utilize the CLT to derive the upper and lower bounds for the optimal error probability. The main results in this paper are based on the following assumption.

\textbf{Assumption A.} (Large system limit) $0<\lim\inf\limits_{N \ge 1}  y_1 \le y_1  \le \lim \sup\limits_{N \ge 1} y_1 <\infty$, $0<\lim\inf\limits_{N \ge 1}  y_2 \le y_2  \le \lim \sup\limits_{N \ge 1} y_2 <\infty$, $0<\lim\inf\limits_{N \ge 1}  \beta \le \beta  \le \lim \sup\limits_{N \ge 1} \beta <\infty$.

\textbf{Assumption A} assumes that $M$, $N$, $n$, and $L$ increase to infinity at the same pace, which is posed to tackle the complex performance evaluation for large-scale MIMO systems. Given $c=\frac{y_1}{y_2}=\frac{N}{M}$, it also guarantees that $0<\lim\inf\limits_{N \ge 1}  c \le c \le \lim \sup\limits_{N \ge 1} c <\infty$. Note that the asymptotic regime in~\textbf{Assumption A} is assumed for the asymptotic analysis but not required for practical operations. This technique has been widely used in evaluating the performance of large MIMO systems~\cite{hachem2008new,moustakas2003mimo,moustakas2023reconfigurable ,zhang2022asymptotic} and the strikingly simple expressions for the asymptotic performance have been validated to be accurate even for low dimensional systems. Furthermore, with SDM, $n$ can be chosen very large, e.g., $64$~\cite{winzer2011mimo, dar2012jacobi} and $N,M$ should be large to achieve high capacity in optical fiber communications.


\subsection{Capacity Analysis}
The ergodic capacity of the Jacobi MIMO channels can be characterized by the following theorem.
\begin{theorem}\label{first_the} Given $0<\lim\inf\limits_{N \ge 1}  y_1 \le y_1  \le \lim \sup\limits_{N \ge 1} y_1 <\infty$ and $0<\lim\inf\limits_{N \ge 1}  y_2 \le y_2  \le \lim \sup\limits_{N \ge 1} y_2 <\infty$, the following evaluation for $C(\sigma^2)=\frac{1}{M}\log\det(\BI_N+\frac{1}{\sigma^2}\BH\BH^{H})$ holds true
\begin{equation}
\label{as_con}
C(\sigma^2)  \xrightarrow[M \xrightarrow{(y_1,y_2)} \infty]{a.s.}  \overline{C}(\sigma^2),
\end{equation}
and
\begin{equation}
\E [C(\sigma^2)]  =  \overline{C}(\sigma^2)+\BO(M^{-2}).
\end{equation}
When $N\le M$, $\overline{C}(\sigma^2)$ is given by 
\begin{equation}
\label{C_exp_1}
\begin{aligned}
\overline{C}(\sigma^2)&=
\log(1+(1+\sigma^{2})\delta)+\frac{1-y_2}{y_2}\log(1+\sigma^2\delta)
\\
&
-\frac{y_1}{y_2}\log(\frac{(1-y_1)\sigma^2\delta}{y_1})
+\frac{\log(1-y_1)}{y_2}
,
\end{aligned}
\end{equation}
where
\begin{equation}
\label{nm_del}
\delta
=\frac{y_1-y_2+(2y_1-1)\sigma^2 +\sqrt{(\sigma^2+\lambda_{-})(\sigma^2+\lambda_{+})}}{2(1-y_1)\sigma^2(1+\sigma^2)}.
\end{equation}
When $N > M$, $\overline{C}(\sigma^2)$ is given by 
\begin{equation}
\label{C_exp_2}
\begin{aligned}
\overline{C}(\sigma^2)&=
\frac{y_1}{y_2}\Biggl[\log(1+(1+\sigma^2)\delta)+\frac{1-y_1}{y_1}\log(1+\sigma^2\delta)
\\
&
-\frac{y_2}{y_1}\log(\frac{(1-y_2)\sigma^2\delta}{y_2})+\frac{\log(1-y_2)}{y_1}\Biggr],
\end{aligned}
\end{equation}
where 
\begin{equation}
\label{mn_del}
\delta
=\frac{y_2-y_1+(2y_2-1)\sigma^2 +\sqrt{(\sigma^2+\lambda_{-})(\sigma^2+\lambda_{+})}}{2(1-y_2)\sigma^2(1+\sigma^2)}.
\end{equation}
Here $ \lambda_{+}$ and $ \lambda_{-}$ are given by
 \begin{equation}
 \begin{aligned}
 \lambda_{+}
&=\Bigl(\sqrt{y_1(1-y_2)}+\sqrt{y_2(1-y_1)}   \Bigr)^2,
\\
\lambda_{-}
&=\Bigl(\sqrt{y_1(1-y_2)}-\sqrt{y_2(1-y_1)}   \Bigr)^2. 
\end{aligned}
\end{equation}
\end{theorem}
\begin{proof} The proof of Theorem~\ref{first_the} is given in Appendix~\ref{proof_first_the}.
\end{proof}
Theorem~\ref{first_the} indicates that when the numbers of channels approach infinity at the same pace, the per-antenna capacity $C(\sigma^2)$ tends to be deterministic, which also occurs in large-scale MIMO wireless systems~\cite{verdu1999spectral}. Moreover, Theorem~\ref{first_the} gives not only the closed-form evaluation for the per-antenna ergodic capacity but also the convergence rate $\BO(M^{-2})$ (approximation accuracy). The convergence rate guarantees that $\sqrt{ML} (\E [C(\sigma^2)] - \overline{C}(\sigma^2))=\BO(M^{-1})$, with which we can replace $\E[C(\sigma^2)]$ with $\overline{C}(\sigma^2)$ such that
\begin{equation}
D(x) \xrightarrow[]{M \xrightarrow{(y_1,y_2,\beta)}\infty}  \Prob\Bigl\{ \sqrt{ML}(I^{\BW,\BH}_{N,M,n,L}(\sigma^2)-\overline{C}(\sigma^2)) \le x  \Bigr\}.
\end{equation}
It is worth noticing that the $\BO(M^{-2})$ convergence rate has also been proved for single-hop~\cite{hachem2008new} and two-hop Rayleigh MIMO channels~\cite{zhang2022asymptotic,zhang2022secrecy,zhuang2024fundamental}.

The approximation for Jacobi MIMO channels in Theorem~(\ref{first_the}) is different from that for Rayleigh MIMO channels~\cite[Eq. (9)]{verdu1999spectral}. Specifically, although the key parameter $\delta$ is the root of a quadratic equation for both cases, the ergodic capacity of Jacobi MIMO channels is also related to $n$. Furthermore, we can prove that~(\ref{C_exp_1}) and~(\ref{C_exp_2}) degenerate to~\cite[Eq. (9)]{verdu1999spectral} when $n  \xrightarrow{(c)} \infty$, which is shown later in Section~\ref{degenerate_sec}.

\subsection{CLT for ID} 
The asymptotic distribution of ID over Jacobi MIMO channels is given by the following theorem.
\begin{theorem}
\label{clt_the}
Given~\textbf{Assumption A} and $\BC_L=\BI_M-\frac{\FS^{(L)}(\FS^{(L)})^{H}}{L}$ with $\FS^{(L)}\in \mathcal{S}_{=}$, the asymptotic distribution of $I^{\BW,\BH}_{N,M,n,L}(\sigma^2)$ converges to a Gaussian distribution. Specifically, there holds true that
\begin{equation}
\sqrt{\frac{ML}{\Xi}}(I^{\BW,\BH}_{N,M,n,L}(\sigma^2)-\overline{C}(\sigma^2))\xrightarrow[M \xrightarrow{(y_1,y_2,\beta)} \infty]{\mathcal{D}} \mathcal{N}(0,1),
\end{equation}
where the asymptotic mean $\overline{C}(\sigma^2)$ is given in~(\ref{C_exp_1}) and~(\ref{C_exp_2}). The asymptotic variance $\Xi$ is given by
\begin{equation}
\Xi=\beta V_1+V_2+\frac{\Tr(\BC_L^2)}{M}\beta V_3,
\end{equation}
where
\begin{equation}
\label{var_exp1}
\begin{aligned}
V_1 &=\log\Biggl(\frac{(\sqrt{\sigma^2+ \lambda_{+}}+\sqrt{\sigma^2+ \lambda_{-}})^{2}}{4\sqrt{(\sigma^2+ \lambda_{+})(\sigma^2+ \lambda_{-})}}\Biggr).
\end{aligned}
\end{equation}
When $N\le M$, $V_2$ and $V_3$ can be represented as 
\begin{align*}
V_2 &=\frac{y_1}{y_2} (1+\frac{1-y_1}{y_1}\sigma^4\delta' ),\numberthis 
\\
V_3 &=\frac{ 1}{ (1+(1+\sigma^2)\delta)^4} \frac{y_2\delta}{ y_1\Bigl(\frac{M(1+\sigma^2) }{N(1+(1+\sigma^2)\delta)^2}+ \frac{N_0\sigma^2  }{N(1+\sigma^2\delta)^2}  \Bigr)},
\end{align*}
where $\delta$ is given in~(\ref{nm_del}).
When $N>M$, $V_2$ and $V_3$ can be represented as 
\begin{align*}
V_2&=(1+\frac{1-y_2}{y_2}\sigma^4\delta'),\numberthis  
\\
V_3
&=\frac{y_1 (1-y_1)\sigma^2\delta^3}{y_2^2 (1+(1+\sigma^2)\delta)^2(1+\sigma^2\delta)}
\\
& \times
\Biggl[1-\frac{1}{(1+\sigma^2\delta)(1+\sigma^2+\frac{\sigma^2(1-y_1)}{y_1}(1+\frac{\delta}{1+\sigma^2\delta})^2 )} \Biggr],
\end{align*}
where $\delta$ is given in~(\ref{mn_del}) and $\delta'=\frac{\mathrm{d} \delta }{\mathrm{d} \sigma^2}$. The convergence of the CDF for ID is given by
 \begin{equation}
 \label{prob_con_rate}
\Prob(\sqrt{\frac{{ML}}{\Xi}}(I_{N,M,n,L}^{\BW,\BH}(\sigma^2)-\overline{C}(\sigma^2)) \le x  )=\Phi(x)+\BO(L^{-\frac{1}{4}}).
\end{equation}
\end{theorem}
\begin{proof} The proof of Theorem~\ref{clt_the} is given in Appendix~\ref{prof_the2}.
\end{proof}
Theorem~\ref{clt_the} indicates that, the ID distribution for any sequence $\BC_{L}$ with $\FS^{(L)}\in \mathcal{S}^{(L)}_{=}$ converges to a normal distribution in the asymptotic regime. Notice that the first term of the asymptotic variance $V_1$ coincides with the asymptotic variance for MI in~\cite[Eq. (56)]{karadimitrakis2014outage}. With this result, we could establish the upper and lower bounds for the optimal average error probability by selecting $\FS^{(L)}$. Different from the CLTs in~\cite{hachem2008new,hoydis2015second,zhang2022outage,zhangasilomar}, we also analyze the error term $\BO(L^{-\frac{1}{4}})$ for the approximation instead of only showing the convergence of the distribution.

\subsection{Comparison with the CLT for Rayleigh Channels}
\label{degenerate_sec}
Now, we compare the CLT for ID over Jacobi channels in~Theorem~\ref{clt_the} with that for Rayleigh channels in~\cite[Theorem 2]{hoydis2015second}. It has been shown that when $n \rightarrow \infty$, Jacobi model degenerates to Rayleigh model, which follows from the intuition that the Wishart ensemble approaches the Jacobi ensemble~\cite{dar2012jacobi}. To compare Jacobi model with Rayleigh model from the perspective of ``randomness'', a power normalization should be performed so that the received SNR of two channels are equal. Specifically, we let $\rho=\frac{n \overline{\rho}}{M}$ ($\rho=\frac{1}{\sigma^2}$ and $\overline{\rho}=\frac{1}{\overline{\sigma}^2}$) such that $\sigma^2=\frac{M}{n}\overline{\sigma}^2$. We will show that for both cases, $N \le M$ and $N > M$, the asymptotic variance $\Xi$ of Jacobi channels will converge to that of Rayleigh channels given in~\cite[Eq. (20)]{hoydis2015second} when $n   \xrightarrow{(c,\beta)} \infty $.
\subsubsection{$N\le M$} In this case, the following convergence for the key parameters hold true
\begin{align*}
\lambda_{+}&\xrightarrow[]{n   \xrightarrow{(c,\beta)} \infty} ( 1+\sqrt{c}  )^2,~~
\lambda_{-} \xrightarrow[]{n   \xrightarrow{(c,\beta)} \infty} ( 1-\sqrt{c} )^2,
\\
\delta & \xrightarrow[]{n   \xrightarrow{(c,\beta)} \infty}\!\! \frac{-(1\!-\!c\!+\!\overline{\sigma}^2)\!+\!\!\sqrt{ ( 1\!-\!c\!+\!\overline{\sigma}^2 )^2\!+\!4c\overline{\sigma}^2   }  }{2\overline{\sigma}^2}\! \!: =\!\delta_0(\overline{\sigma}^2),
\\
\delta' & \xrightarrow[]{n   \xrightarrow{(c,\beta)} \infty} y_2^{-1}\delta'_{0}(\overline{\sigma}^2) ,   \numberthis
\\
V_2 &
\xrightarrow[]{n   \xrightarrow{(c,\beta)} \infty} c+\overline{\sigma}^4\delta'_0(\overline{\sigma}^2)  ,
\\
V_3 
&
\xrightarrow[]{n   \xrightarrow{(c,\beta)} \infty} \!\! \frac{\beta\delta_0(\overline{\sigma}^2)}{(1\!+\!\delta_0(\overline{\sigma}^2))^4(\overline{\sigma}^2 \!+\! \frac{1}{(1+\delta_0^2(\overline{\sigma}^2))^2}) }
\!=\! \frac{-\beta\delta_0'(\overline{\sigma}^2)}{(1\!+\!\delta_0(\overline{\sigma}^2))^4 }.
\end{align*}
This guarantees that $\overline{C}(\sigma^2)$ and $\Xi$ converge to~\cite[Eqs. (12) and (20)]{hoydis2015second}, respectively.
\subsubsection{$N> M$} In this case, $\delta$ converges to $\widetilde{\delta_0}(\overline{\sigma}^2)=\delta_0(\overline{\sigma}^2)+\frac{M-N}{Nz}$. The convergence of the asymptotic mean and variance can be given by
\begin{align*}
\delta & \xrightarrow[]{n   \xrightarrow{(c,\beta)} \infty} \frac{-(c-1+\overline{\sigma}^2)+\sqrt{ ( c-1+\overline{\sigma}^2 )^2+4c\overline{\sigma}^2   }  }{2\overline{\sigma}^2}
\\
&
  :=\widetilde{\delta_0}(\overline{\sigma}^2)=\delta_0(\overline{\sigma}^2)+\frac{M-N}{N\overline{\sigma}^2},
\\
V_3&   \xrightarrow[]{n   \xrightarrow{(c,\beta)} \infty}   \numberthis
\frac{1}{(1+\delta_0(\overline{\sigma}^2))^3}\Biggl[\frac{M\delta_0^2(\overline{\sigma}^2)}{N} 
\\
&
- \frac{M\overline{\sigma}^2\delta_0^3(\overline{\sigma}^2)}{N(1+\delta_0(\overline{\sigma}^2))(\overline{\sigma}^2+\frac{1}{(1+\delta_0(\overline{\sigma}^2))^2})} \Biggr]
\\
&
=\frac{1}{(1+\delta_0(\overline{\sigma}^2))^4}\frac{\delta_0(\overline{\sigma}^2)}{\overline{\sigma}^2+\frac{1}{(1+\delta_0(\overline{\sigma}^2))^2}}.
\end{align*}
Different from the result for Rayleigh channels, the CLT of Jacobi channels also depends on the number of available channels $n$. When $n$ grows larger, the dependence between the entries of $\BH$ becomes weaker and the channel approaches the Rayleigh case, whose channel coefficients are independent of each other. This agrees with the analysis in~\cite{dar2012jacobi}.

\subsection{Upper and Lower Bounds for Error Probability} 
With Theorem~\ref{clt_the}, the closed-form upper and lower bounds for optimal error probability with a rate close to capacity are given by the following theorem.
\begin{theorem} 
\label{the_oaep}
For a given rate $R=\overline{C}(\sigma^2)+\frac{r}{\sqrt{ML}}$, the optimal average error probability $\Prob_{\mathrm{e}}(r|y_1,y_2,\beta)$ of Jacobi MIMO channels is lower and upper bounded by
\begin{equation}\label{lower_exp}
\Prob_{\mathrm{e}}(r|  y_1,y_2,\beta) \ge 
\begin{cases}
\Phi\Bigl(\frac{r}{\sqrt{\Xi_{-}}}\Bigr)+\BO(L^{-\frac{1}{2}}),~~ r\le 0,
\\
\frac{1}{2},~~ r> 0,
\end{cases}
\end{equation}
and
\begin{equation}\label{upper_exp}
\Prob_{\mathrm{e}}(r|y_1,y_2,\beta)  \le \Phi\Bigl(\frac{r}{\sqrt{\Xi_{+}}}\Bigr)+\BO(L^{-\frac{1}{2}}),
\end{equation}
respectively, where
\begin{equation}
\begin{aligned}
\label{var_upp_low}
\Xi_{-}&=\beta V_1+V_2,~~
\\
\Xi_{+}&=\beta V_1+V_2+V_3.
\end{aligned}
\end{equation}
\end{theorem}
\begin{proof}
Theorem~\ref{the_oaep} can be proved based on Theorem~\ref{clt_the}, which is similar to~\cite[Theorem 3]{hoydis2015second} and~\cite[Theorem 3]{zhang2022second}, and omitted here. 
\end{proof}

In the following remarks, we compare the upper and lower bounds with existing results.

\setcounter{remark}{1}
\begin{subremark}
\label{rem_rayleigh}
\textbf{Comparison of the upper and lower bounds between Jacobi and Rayleigh MIMO channels:}
As shown in Section~\ref{degenerate_sec}, the variance terms $\beta V_1$, $V_3$, and $V_2$ will degenerate to the first, last term, and the sum of the second and third terms in the dispersion~\cite[Eq. (25)]{hoydis2015second}, respectively, when $\sigma^2=\frac{M}{n}\overline{\sigma}^2$ with $n   \xrightarrow{(c,\beta)} \infty$. This indicates that the bounds for Jacobi channels degenerate to those for Rayleigh channels.
\end{subremark}

\begin{subremark}
\label{rem_outage}
 \textbf{Comparison of the upper and lower bounds with outage probability:} The outage probability with IBL can be obtained by letting $\beta   \xrightarrow[]{(y_1,y_2)}\infty$ ($L$ has a higher order than the number of channels $M$, $N$, and $n$, whose ratios are fixed). Specifically, we have 
\begin{equation}
\begin{aligned}
\label{outage_exp}
&\Prob_{\mathrm{e}}(r|y_1, y_2,\beta)\xlongrightarrow[]{\beta   \xrightarrow[]{(y_1,y_2)}\infty} P_{\mathrm{out}}(R)
\\
&=\Phi\Bigl(\frac{M(R-\overline{C}(\sigma^2))}{\sqrt{V_1}}\Bigr)+\BO(N^{-\frac{1}{2}}), 
\end{aligned}
\end{equation}
which agrees with~\cite[Eq. (59)]{karadimitrakis2014outage} and the variance $V_1$ is equal to~\cite[Eq. (56)]{karadimitrakis2014outage}. The outage probability in~\cite{karadimitrakis2014outage} was derived by large deviation method. Compared with the asymptotic RMT, the large deviation method achieves higher accuracy when $R$ is much smaller than the capacity, but with complex expressions. Furthermore, the large deviation results agree with the asymptotic RMT results in~(\ref{outage_exp}) when the rate is close to the capacity. Different from the outage probability, the impact of FBL is reflected by $\beta$, $V_2$, and $V_3$ in~(\ref{var_upp_low}).
\end{subremark}

\begin{subremark}
\textbf{Comparison of the error exponent between Jacobi and Rayleigh MIMO channels: } The exponent of error probability (Gallagar bound) for Jacobi MIMO channels was investigated in~\cite{opkaradimitrakis2017gallager} in the asymptotic regime. When the rate is close to capacity, the error exponent for the case $M+N=n$, and $N \le M$ ($\lambda_{+}=1$, $\lambda_{-}=(y_2-y_1)^2$) is given by
\begin{equation}
\label{expo_jaco}
E_{G}=V_1+\frac{(\frac{y_1}{y_2}+1)(1-\sqrt{\frac{\sigma^2+\lambda_{-}}{\sigma^2+1}})}{\beta }.
\end{equation}
In this case, the variance $\Xi_{+}$ in~(\ref{var_upp_low}) can be represented by
\begin{equation}
\label{gala_degene}
\frac{\Xi_{+}}{\beta}=V_1+\frac{2\omega-\omega^2}{\beta},
\end{equation}
where 
\begin{align*}
&\omega\!\!=\!\!\frac{N\delta}{M(1+(1+\sigma^2)\delta)}
\!\!=\!\!\frac{N\bigl(y_2-y_1 +\sqrt{\frac{\sigma^2+\lambda_{-}}{\sigma^2+1}}\bigr)}{M\Bigl( \sigma^2\!+\!y_2 \!-\! y_1 \!+\!(1+\sigma^2)\sqrt{\frac{\sigma^2\!+\!\lambda_{-}}{\sigma^2+1}} \Bigr)}
\\
& \!\!\overset{(a)}{=}\!\!\frac{N\Bigl( -2y_1\sigma^2+2y_1\sigma^2\sqrt{\frac{\sigma^2+\lambda_{-}}{\sigma^2+1}}\Bigr)}{-4M\sigma^2y_1^2}\!\!
=\!\!\frac{(1+\frac{y_1}{y_2}) \Bigl(1-\sqrt{\frac{\sigma^2+\lambda_{-}}{\sigma^2+1}}\Bigr)}{2}.\numberthis \label{gala_degene1}
\end{align*}
Step $(a)$ in~(\ref{gala_degene1}) follows by $2(y_2 - y_1)=2(y_2^2 -y_1^2)$. By comparing~(\ref{expo_jaco}) and~(\ref{gala_degene}), we have $\frac{\Xi_{+}}{\beta}=E_G-\frac{\omega^2}{\beta}<E_G$. This indicates that the upper bound in Theorem~\ref{the_oaep}, which is achieved by spherical Gaussian codebook~\cite[Theorem 3]{hoydis2015second}, is tighter than the error exponent when the rate is close to the capacity. Furthermore, Theorem~\ref{the_oaep} is also valid when $N \le M$ while the error exponent in~\cite{opkaradimitrakis2017gallager} is only derived for $N>M$.

Theorem~\ref{the_oaep} is more powerful in characterizing the error probability than the error exponents in~\cite{karadimitrakis2017gallager} when the rate is close to the capacity since the error probability is represented by a Gaussian approximation with closed-form mean and variance. The asymptotic Gaussianity has been strictly proved in Theorem~\ref{clt_the}. The correctness of this bound can be further validated by the fact that it coincides with the error exponent for Rayleigh MIMO channels when $n\rightarrow \infty$ with $\sigma^2=\frac{M\overline{\sigma}^2}{n}$, as shown below. In~\cite{karadimitrakis2017gallager}, it has been shown that for Rayleigh MIMO channels, the error exponent saturates the dispersion in upper bound~\cite[Eq. (25)]{hoydis2015second} when the rate is close to capacity. The dispersion $\Xi_{+}$ in~(\ref{gala_degene}) of Jacobi MIMO channels can degenerate to the error exponent of Rayleigh MIMO channels. Specifically, we have
\begin{align*}
\omega &\xrightarrow[]{n   \xrightarrow{(c,\beta)} \infty} \frac{\delta_0(\overline{\sigma}^2)}{1+\delta_0(\overline{\sigma}^2)}  \numberthis
\\
&=\frac{\left(\sqrt{\overline{\sigma}^2+(1+\sqrt{\frac{y_1}{y_2}})^2}+\sqrt{\overline{\sigma}^2+(1-\sqrt{\frac{y_1}{y_2}})^2}\right)^2}{4},
\end{align*}
which agrees with $g_0$ and the error exponent given in~\cite[Eqs.(31) and (32)]{karadimitrakis2017gallager}. This indicates that the Gallager random
coding exponent with Gaussian input saturates the derived dispersion in the upper bound when $n   \xrightarrow{(c,\beta)} \infty$. 
\end{subremark}








The high SNR approximation for the bounds in Theorem~\ref{the_oaep} can be obtained by taking $\sigma^2 \rightarrow 0$, which is given in the following proposition.
\begin{proposition} 
\label{high_snr_pro}
(High SNR approximation) When $\sigma^2 \rightarrow 0$, $\overline{C}(\sigma^2)$, $\Xi_{-}$, and $\Xi_{+}$ can be approximated by
\begin{equation}
\overline{C}(\sigma^2)=\frac{\min\{y_1,y_2 \}}{y_2} +\BO(1),
\end{equation}
\begin{equation}
\label{v__approx-}
\Xi_{-}=
\begin{cases}
\begin{aligned}
&\beta V^{(\infty,1,2)}_1+\frac{y_1}{y_2}+\BO(\sigma^2),~N< M,
\\
&  \beta V^{(\infty,2,1)}_1  +1+\BO(\sigma^2),~M>N,
\\
&\BO(\sigma^{-1}),~M=N, 
\end{aligned}
\end{cases}
\end{equation}
\begin{equation}
\label{v__approx+}
\begin{aligned}
\Xi_{+}=
\begin{cases}
\begin{aligned}
&\beta V^{(\infty,1,2)}_1+\frac{2y_1}{y_2}-\frac{y_1^2}{y_2^2}+\BO(\sigma^2),~N<M,
\\
 & \beta V^{(\infty,2,1)}_1+\BO(\sigma^2), ~N>M,
\\
&\BO(\sigma^{-1}),~N=M,
\end{aligned}
\end{cases}
\end{aligned}
\end{equation}
respectively, where $V^{(\infty,i,j)}_1=-\log\Bigl( 1- \frac{y_i(1-y_j)}{y_j(1-y_i)} \Bigr)$ with $i,j \in \{1,2 \}$.
\end{proposition}
\begin{remark}
\label{high_rem}
It is worth noticing that the dominating term of $\overline{C}(\sigma^2)$ has a coefficient that is related to the ratio between $\min \{M,N \}$ and $N$, and is independent of $n$, which indicates that in the high SNR regime, the capacity is limited by the minimum of the number of transmit and receive channels. This coincides with the high SNR approximation for Rayleigh MIMO channels in~\cite[Eqs. (36) and (37)]{zhang2022asymptotic}. In fact, $n$ reflects the dependence between transmit and receive channels. In the high SNR regime, the dependence between the channel coefficients has much less effect on capacity than the number of transmit and receive channels. Besides the $\log$ term, terms in $\Xi_{-}$ and $\Xi_{+}$ are the same as those for Rayleigh MIMO channels. When $n$ grows larger for $N<M$ and $N>M$, the term $V^{(\infty,i,j)}_1=-\log\Bigl( 1- \frac{y_i(1-y_j)}{y_j(1-y_i)} \Bigr)$ in~(\ref{v__approx-}) and~(\ref{v__approx+}), which represents the asymptotic variance for the MI in~(\ref{outage_exp}), increases. This results in a larger error probability. Furthermore, when $N\le M$, we have $\Xi_{-} \xrightarrow[]{n   \xrightarrow{(c,\beta)} \infty}-\beta\log(1-c)+c$ and $\Xi_{+} \xrightarrow[]{n   \xrightarrow{(c,\beta)} \infty}=-\beta\log(1-c)+c(2-c)$, which agree with the high SNR approximations for Rayleigh MIMO channels in~\cite[Eqs. (27) and (28)]{hoydis2015second}.
\end{remark}


\section{Numerical Results}
\label{sec_simu}

In this section, we validate the derived theoretical results by numerical simulations.
\begin{figure}
\vspace{-0.4cm}
\centering
\includegraphics[width=0.4\textwidth]{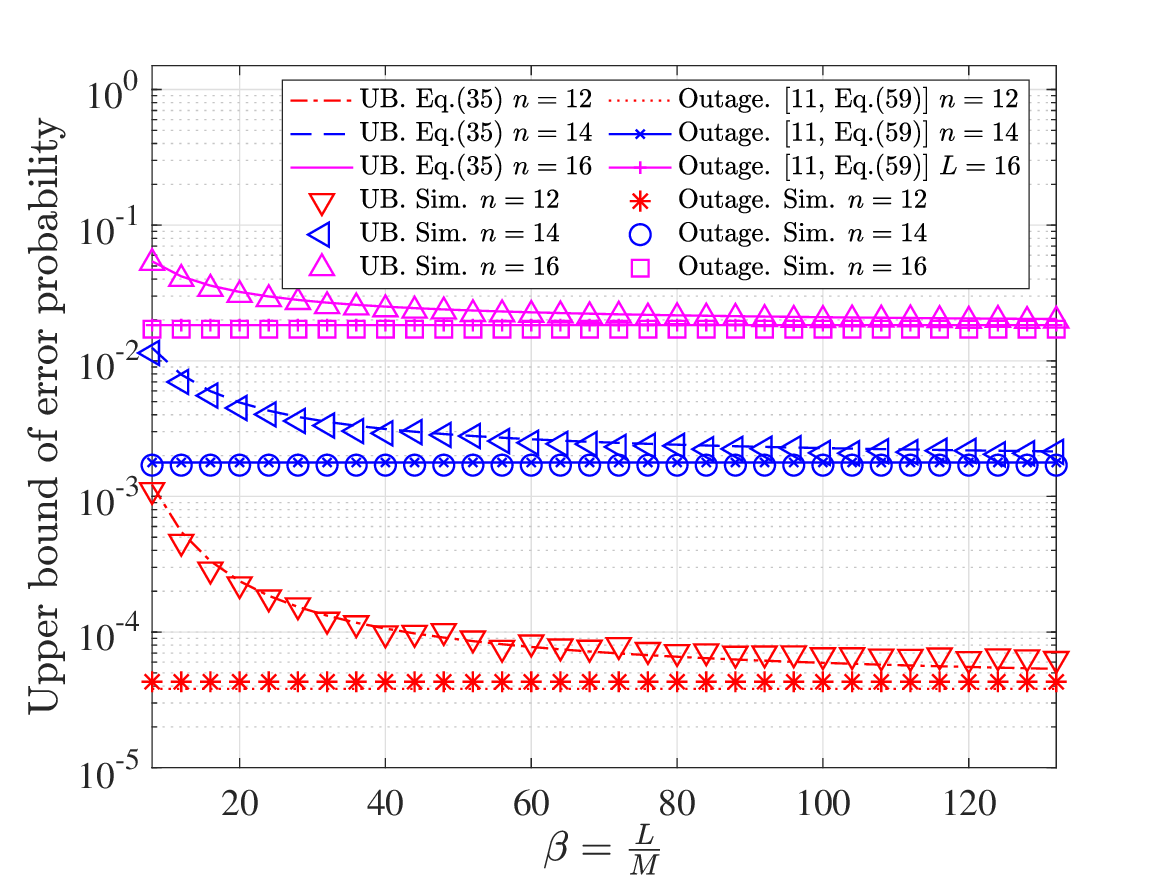}
\captionof{figure}{Approximation accuracy of the derived bounds.}
\label{simu_up}
\vspace{-0.5cm}
\end{figure}
\begin{figure}
\centering
\includegraphics[width=0.4\textwidth]{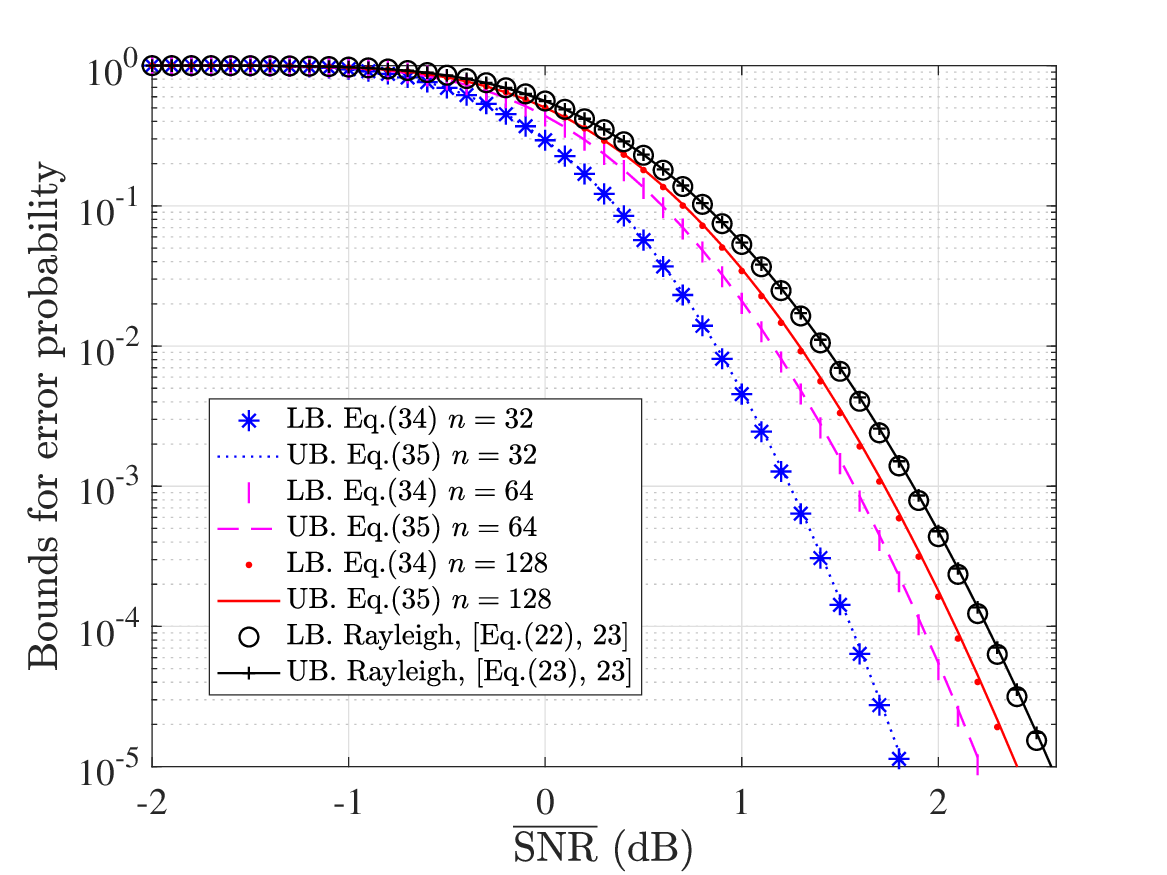}
\captionof{figure}{Bounds with different $n$.}
\label{bounds_n}
\vspace{-0.5cm}
\end{figure}
\begin{figure}
\vspace{-0.5cm}
\centering
\includegraphics[width=0.4\textwidth]{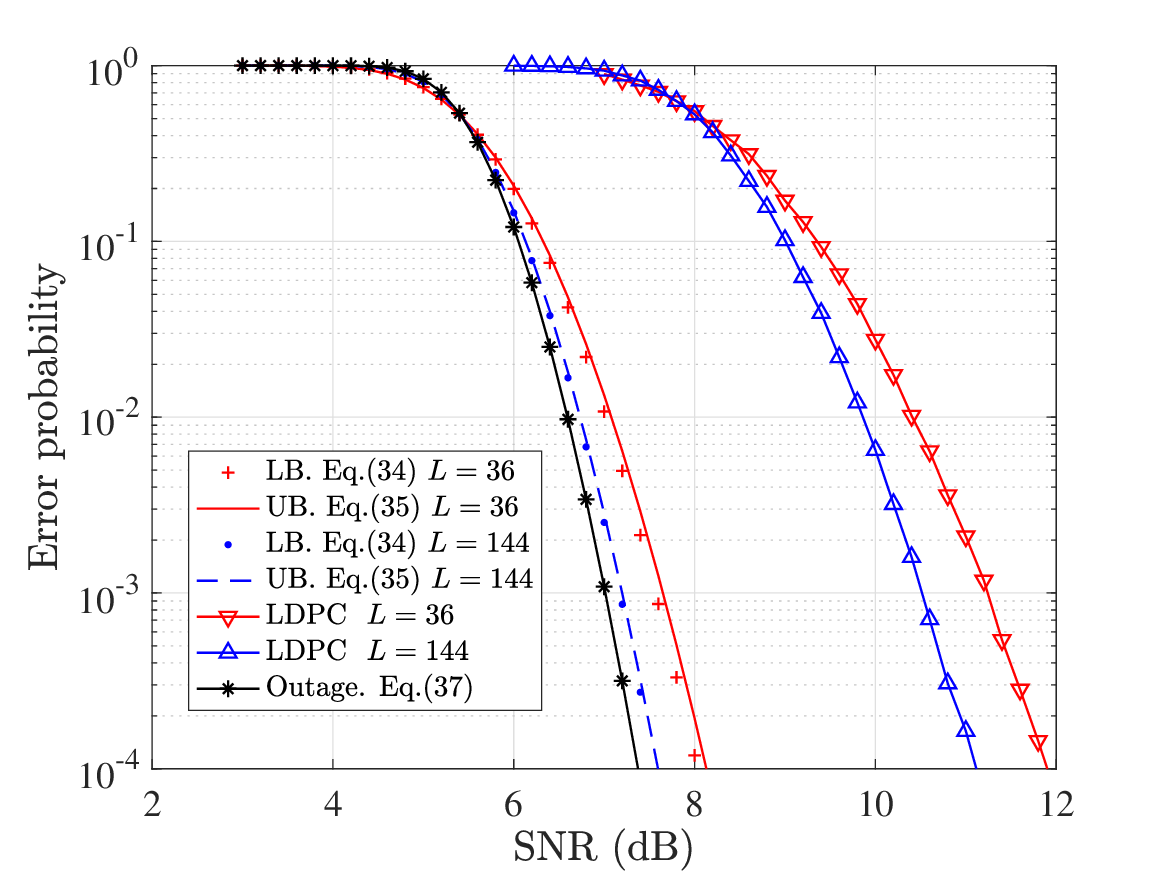}
\captionof{figure}{Bounds and LDPC codes ($n=16$).}
\label{simu_ldpc}
\end{figure}

Fig.~\ref{simu_up} depicts the theoretical upper bound in~(\ref{upper_exp}) and the simulation values are generated by $10^{7}$ Monte-Carlo realizations (noted as Sim.). The simulation settings are: $M=6$, $N=4$, $\sigma^{-2}= 5$ dB, $n=\{12,14,16 \}$, and $R=0.37$ nat/s/Hz. It can be observed that the analytical expressions in Theorem~\ref{the_oaep} are accurate. Meanwhile, the outage probability is too optimistic for small $\beta$. As $\beta$ grows larger, the upper bound approaches outage probability, which validates the analysis in Remark~\ref{rem_outage}. As will be shown in Fig.~\ref{bounds_n}, the gap between the upper and lower bound is very small. Thus, for a better illustration, we omit the lower bound in Fig.~\ref{simu_up}.

Fig.~\ref{bounds_n} compares the derived bounds for Jacobi channels and those for Rayleigh channels with normalized receive SNR $\overline{\mathrm{SNR}}=\overline{\sigma}^{-2}$. The settings are: $M=8$, $N=16$, $n=\{32,64,128 \}$, $L=36$, and $R=1$ nat/s/Hz. It can be observed that as $n$ increases, both the upper and lower bounds for Jacobi channels increase, which agrees with the high SNR analysis in Remark~\ref{high_rem}. Meanwhile, the gap between the upper and lower bounds is small. Furthermore, the bounds for Jacobi channels approach those for Rayleigh channels as the number of available channels increases, which validates the analysis in~Section~\ref{degenerate_sec} and Remark~\ref{rem_rayleigh}.

Next, we compare the performance of specific coding schemes with the derived upper and lower bounds. Here, the $1 \slash 2$  LDPC code is considered for optical fiber communications~\cite{zhang2017low}, where $M=8$, $N=6$, and $n=16$. At the transmitter side, the bit interleaved coded modulation and QPSK modulation are adopted such that the rate is $R=\log(2)$  nat/s/Hz.  Two codelengths of $l \in \{576,2304\}$ bits are considered and the corresponding blocklengths are $L=\frac{l}{2M}\in\{36, 144 \} $. At the receiver side, the maximal likelihood demodulator~\cite{muller2002coding,mckay2005capacity} is adopted such that the log likelihood ratio of the $i$-th bit $s_{i}\in \{0,1\}$ is given by
\begin{equation}
L(s_{i}|\bold{R}^{(L)},\BH)=\log\frac{\sum_{\bold{c} \in \mathcal{C}^{(i)}_{1} } p(\bold{R}^{(L)}|\bold{c},\BH) }{\sum_{\bold{c}\in \mathcal{C}^{(i)}_{0} } p(\bold{R}^{(L)}|\bold{c},\BH) }.
\end{equation}
Here $\mathcal{C}^{(i)}_{1}=\{\bold{c}| c_{i}=1,\bold{c}\in \mathcal{C}   \}$ consists of codewords whose $i$-th digit is $1$ and $\mathcal{C}^{(i)}_{0}$ consists of codewords whose $i$-th digit is $0$. $p(\bold{R}^{(L)}|\bold{c},\BH)$ denotes the conditional probability density function given $\BH$ and $\bold{c}$. The demodulation  results are then passed to the LDPC decoder for soft decision. Under such circumstances, the second-order coding rate is given by $r=\sqrt{ML}(R-\overline{C}(\sigma^2))$. In Fig.~\ref{simu_ldpc}, the error probability of the above LDPC codes is compared with the derived bounds in Theorem~\ref{the_oaep}. It can be observed that outage probability is too optimistic while the derived bounds are closer to the LDPC performance with practical blocklength. Furthermore, outage probability decays very fast when SNR increases, while the slope of the bounds matches that of the error probability for LDPC codes and the performance gap is around $3$ dB. 

\section{Conclusion}
\label{sec_con}
To meet the stringent latency requirement of future communication systems, this paper studied the optimal average error probability of optical fiber multicore/multimode communications  with FBL when the coding rate is a perturbation within $\BO(\frac{1}{\sqrt{ML}})$ of the capacity. We derived the CLT for ID of Jacobi MIMO channels in the asymptotic regime where the number of transmit, receive, and available channels go to infinity at the same pace. The result was then utilized to derive the upper and lower bounds for the optimal average error probability. Compared with outage probability derived in the IBL regime, the bounds can characterize the impact of the blocklengh and degenerate to outage probability when the blocklength approaches infinity. Moreover, the bounds approach those of Rayleigh MIMO channels when the number of channels approaches infinity. Numerical results validated the accuracy of the bounds and showed that outage probability is too optimistic and the derived bounds are closer to the error probability of practical LDPC coding schemes, indicating that the derived bounds provide a better evaluation for optical fiber multicore/multimode communications.

\appendices

\section{Proof of Theorem~\ref{first_the}}
\label{proof_first_the}
The convergence of the $\log\det$ term resorts to the convergence of trace of a specific inverse matrix, $\E[\Tr(\BG(a,b))]$, whose evaluation is based on the following lemma.
\begin{lemma} 
\label{trace_lemma}
Define $\BG(a,b)=\left(a\BX\BX^{H}+b\BY\BY^{H} \right)^{-1}$ with $a>b>0$, where $\BX \in \mathbb{C}^{N \times M_1}$ and $\BY \in \mathbb{C}^{N \times M_2 }$ with $N\le M_1$ and $M_2>0$ are i.i.d. Gaussian random matrices, whose entries have zero mean and $\frac{1}{N}$ variance. When $N \xrightarrow{(c_1,c_2)} \infty  $, with $c_1=\frac{M_1}{N}$ and $c_2=\frac{M_2}{N}$, there holds true that
\begin{equation}
\label{resol_rate}
\frac{1}{N}\E[\Tr(\BG(a,b))]=\delta(a,b)+\BO(\frac{1}{N^2}),
\end{equation}
and
\begin{equation}
\begin{aligned}
\label{second_resol}
&\frac{1}{N}\E [\Tr(\BG(a,b)\BG(c,d))]
\\
&
=\frac{\delta(a,b)}{ \frac{\frac{c M_1}{N}}{(1+c\delta(c,d))(1+a\delta(a,b))}+ \frac{\frac{d M_2}{N}}{(1+d\delta(c,d))(1+b\delta(a,b))}  }+\BO(\frac{1}{N^2}),
\end{aligned}
\end{equation}
where $\delta(a,b)$ is the solution of the following equation
\begin{equation}
\label{basic_eq}
ab(c_1+c_2-1)\delta^2+(ac_1+bc_2-a-b)\delta-1=0.
\end{equation}
\end{lemma}
\begin{proof} 
For brevity, we denote $\alpha_{G}=\frac{\E[\Tr(\BG(a,b))]}{N}$. The proof includes two steps. In the first step, we show that the deterministic approximation for $\alpha_G$, i.e., $\delta$, is the solution of the quadratic equation in~(\ref{basic_eq}) utilizing the Gaussian tools, i.e., integration by parts formula~\cite[Eq. (40)]{zhang2022asymptotic} and~Poincar{\`e}-Nash inequality~\cite[Eq. (18)]{hachem2008new}. In the second step, we show that the approximation error of $\delta$ is $\BO(N^{-2})$.

\subsection{Step 1: Deterministic approximation for $\alpha_{G}$}
To evaluate $\alpha_{G}$, we first define $\alpha_{GX}=\frac{a}{N}\E[\Tr(\BG(a,b)\BX\BX^{H})]$ and $\alpha_{GY}=\frac{b}{N}\E[\Tr(\BG(a,b)\BY\BY^{H})]$. By the integration by parts formula~\cite[Eq. (17)]{hachem2008new}, we can obtain
\vspace{-0.1cm}
\begin{equation}
\label{trace_gxx}
\begin{aligned}
&\alpha_{GX}
=\!\frac{a}{N} \sum_{i,j} \E [X_{i,j}^{*} [\BG(a,b)\BX]_{i,j}]
\\
&=\!\frac{a c_1 \E[\Tr(\BG(a,b))]}{N}
\!-\!\E \Bigl[\! \frac{a\Tr(\BG(a,b))}{N} \frac{a\Tr(\BG(a,b)\BX\BX^{H})}{N} \! \Bigr]
\\
&=\!\frac{ac_1\E[\Tr(\BG(a,b))]}{N}
\!-\! \frac{a\E[\Tr(\BG(a,b))]}{N} 
\\
&
\times
\frac{\E[a\Tr(\BG(a,b)\BX\BX^{H})]}{N} +\varepsilon_1, 
\end{aligned}
\vspace{-0.3cm}
\end{equation}
where $\varepsilon_1$ can be evaluated by
\begin{equation}
\begin{aligned}
\label{cov_ep1}
\varepsilon_1&=a^2 N^{-2}\cov(\Tr(\BG(a,b)\BX\BX^{H}), \Tr(\BG(a,b))) 
\\
&
\le a^2 N^{-2} \Var^{\frac{1}{2}}(\Tr(\BG(a,b)\BX\BX^{H}))\Var^{\frac{1}{2}}(\Tr(\BG(a,b))).
\end{aligned}
\end{equation}
By the Nash-Poincar{\'e} inequality~\cite[Eq. (18)]{hachem2008new}, we have the following bound
\begin{align*}
&\Var(\Tr(\BG(a,b)\BX\BX^{H}))
\\
&\le  \sum_{i,j}\frac{1}{N} \E \Bigl[\Bigl|\frac{\partial \Tr(\BG(a,b)\BX\BX^{H})}{\partial X_{i,j}}\Bigr|^2 \Bigr]
\\
&
+ \sum_{i,j}\frac{1}{N} \E\Bigl[ \Bigl|\frac{\partial \Tr(\BG(a,b)\BX\BX^{H})}{\partial X_{i,j}^{*}}\Bigr|^2\Bigr] 
\\
&+
\sum_{k,l}  \frac{1}{N} \E \Bigl[\Bigl|\frac{\partial \Tr(\BG(a,b)\BX\BX^{H})}{\partial Y_{k,l}}\Bigr|^2\Bigr] \numberthis
\\
&
+ \sum_{k,l}\frac{1}{N} \E \Bigl[ \Bigl|\frac{\partial \Tr(\BG(a,b)\BX\BX^{H})}{\partial Y_{k,l}^{*}}\Bigr|^2\Bigr]
\\
&:=A_1+A_2+A_3+A_4.
\end{align*}
The term $A_1$ can be evaluated by
\begin{align*}
A_1 &\le \frac{2\E\Tr(\BG^2(a,b)\BX\BX^{H})}{N}
\\
&
+\frac{2a^2\E\Tr((\BG(a,b)\BX\BX^{H})^2\BX\BX^{H})}{N}
\\
& \le \frac{2}{N} \E \| \BG(a,b)  \|  \Tr( \BG(a,b)\BX\BX^{H}) \numberthis
\\
&
+\frac{2\E[\Tr(\BX\BX^{H})]}{N}
  \le \frac{p M}{a^2 N}=\BO(1).
\end{align*}
With the same approach, we can obtain $A_2=\BO(1)$, $A_3=\BO(1)$, and $A_4=\BO(1)$ such that~(\ref{cov_ep1}) is bounded by 
\begin{align}
\cov(\Tr(\BG(a,b)\BX\BX^{H}), \Tr(\BG(a,b)))=\BO(1)
\end{align}
and $\varepsilon_1$ in~(\ref{trace_gxx}) is $\BO(N^{-2})$. Therefore, we can rewrite~(\ref{trace_gxx}) as
\begin{equation}
\alpha_{GX}=a c_1\alpha_G-a\alpha_G\alpha_{GX}+\BO(\frac{1}{N^2}).
\end{equation}
Similarly, we have the following evaluation
\begin{equation}
\label{GY_eva}
\alpha_{GY}
=b c_2 \alpha_G-b\alpha_G\alpha_{GY}+\BO(\frac{1}{N^2}).
\end{equation}
Noticing that $\alpha_{GX}+\alpha_{GY}=\frac{1}{N}\E[\Tr(\BG(a,b)\BG^{-1}(a,b))]=1$, we can obtain that $\alpha_G$ is the solution of the following equation
\begin{equation}
\label{alpha_eq}
\frac{a c_1\alpha_G}{1+a \alpha_G}+\frac{bc_2\alpha_G}{1+b \alpha_G}=1+\BO(\frac{1}{N^2}).
\end{equation}
This indicates that the approximation for $\alpha_G$ should satisfy an equation similar to~(\ref{alpha_eq}). Thus, we define $\delta=\delta(a,b)$ as the positive solution of the following quadratic equation
\begin{equation}
\label{delta_eq_ab}
\frac{a c_1\delta(a,b)}{1+a \delta(a,b)}+\frac{bc_2\delta(a,b)}{1+b \delta(a,b)}=1,
\end{equation}
which is the deterministic approximation for $\alpha_{G}$. Note that ~(\ref{delta_eq_ab}) is equivalent to~(\ref{basic_eq}).

\subsection{Step 2: Convergence rate of $\delta$}
Now, we will show that $\alpha_{G}=\delta+\BO(\frac{1}{N^2})$. By computing the difference between~(\ref{alpha_eq}) and~(\ref{delta_eq_ab}), we have
\begin{align*}
&\delta-\alpha_G
\\
&
=\frac{\frac{a^2c_1}{(1+a\delta)(1+a\alpha_G)}+\frac{b^2c_2}{(1+b\delta)(1+b\alpha_G)}}{(\frac{ac_1}{1+a \delta}+\frac{bc_2}{1+b \delta})(\frac{ac_1}{1+a c_1\alpha_G}+\frac{bc_2}{1+b c_2\alpha_G})}(\delta-\alpha_G)+\BO(\frac{1}{N^2}) \numberthis
\\
&=\frac{\frac{a^2c_1\delta}{(1+a\delta)(1+a\alpha_G)}+\frac{bc_2\delta}{(1+b\delta)(1+b\alpha_G)}}{\frac{ac_1}{1+a \alpha_G}+\frac{bc_2}{1+b \alpha_G}}(\delta-\alpha_G)+\BO(\frac{1}{N^2})
\\
&=C_{a,b}(\delta-\alpha_G)+\BO(\frac{1}{N^2}).
\end{align*}
Thus, $\alpha_{G}=\delta+\BO(\frac{1}{n^2})$ can be obtained by proving $C_{a,b}<1$, which is shown as follows
\begin{equation}
\begin{aligned}
\label{Cab_1}
C_{a,b}
&=\frac{(1-\frac{c_1}{1+a\delta})\frac{a}{1+a\alpha_G}+(1-\frac{c_2}{1+b\delta})\frac{b}{1+b\alpha_G}}{\frac{ac_1}{1+a \alpha_G}+\frac{bc_2}{1+b \alpha_G}}
\\
&=1-\frac{\frac{c_1}{1+a\delta}\frac{a}{1+a\alpha_G}+\frac{c_2}{1+b\delta}\frac{b}{1+b\alpha_G}}{\frac{a}{1+a \alpha_G}+\frac{b}{1+b \alpha_G}}
=1-D_{a,b}.
\end{aligned}
\end{equation}
Then we only need to show that $D_{a,b}$ is bounded away from zero, which is achieved by analyzing the bounds for $\delta$ and $\alpha_G$. By~(\ref{delta_eq_ab}) and $a>b>0$, we have
\begin{equation}
\frac{(c_1+c_2)b \delta}{1+b\delta}<\frac{a c_1\delta}{1+a \delta}+\frac{bc_2\delta}{1+b \delta}=1,
\end{equation}
such that 
\begin{equation}
\label{delta_bnd}
\delta<\frac{1}{(c_1+c_2)b}.
\end{equation}
Moreover, we have 
\begin{equation}
\label{bnd_alphaG}
\begin{aligned}
&[a(1+\sqrt{c_1})^2+b(1+\sqrt{c_2})^2]^{-1}
\\
<&\alpha_G<\frac{\E[\Tr(\BX\BX^{H})^{-1}]}{Na}\le a^{-1}(1-\sqrt{c_1})^{-2},
\end{aligned}
\end{equation}
which follows from the fact that the eigenvalues of $\BX\BX^{H}$ locate in the interval $[1-c_1^{-\frac{1}{2}}, 1+c_1^{-\frac{1}{2}}]$ with~\textbf{Assumption A}. By~(\ref{Cab_1}),~(\ref{delta_bnd}) and the bounds for $\alpha_G$ in~(\ref{bnd_alphaG}), we can conclude that there exists a constant $K$ independent of $N$, $M_1$, and $M_2$ such that $D_{a,b}>K>0$. Therefore, we have $C_{a,b}<1$ such that $\alpha_{G}=\delta+\BO(\frac{1}{N^2})$, which concludes~(\ref{resol_rate}).

Now we turn to prove~(\ref{second_resol}). By the integration by parts formula~\cite[Eq. (17)]{hachem2008new}, we have
\begin{equation}
\begin{aligned}
\label{GG_1}
&\frac{1}{N}\E[\Tr(\BG(a,b)\BX\BX^{H}\BG(c,d))]
\\
&=\frac{1}{N}\sum_{i,j} \E [X^{*}_{i,j} [\BG(c,d)\BG(a,b)\BX]_{i,j}] 
\\
&=\frac{1}{N}\E\Bigl[\frac{\partial [\BG(c,d)\BG(a,b)\BX]_{i,j}}{\partial X_{i,j}}\Bigr]
\\
&
=\frac{\frac{M_1}{N^2}\E[\Tr(\BG(a,b)\BG(c,d))]}{(1+c\delta(c,d))(1+a\delta(a,b))}
+\BO(N^{-2}),
\end{aligned}
\end{equation}
and
\begin{equation}
\begin{aligned}
\label{GG_2}
&\frac{1}{N}\E[\Tr(\BG(a,b)\BY\BY^{H}\BG(c,d))]
\\
&
=\frac{\frac{M_2}{N^2}\E[\Tr(\BG(a,b)\BG(c,d))]}{(1+d\delta(c,d))(1+ b\delta(a,b))}
+\BO(N^{-2}).
\end{aligned}
\end{equation}
By $c\times(\ref{GG_1})+ d\times(\ref{GG_2})$, we can conclude~(\ref{second_resol}).
\end{proof}
Now we begin to prove Theorem~\ref{first_the}. According to~\cite[Eq. (4)]{zhang2021bias}, $C(\sigma^2)$ can be represented as the following integral,
\begin{equation}
\begin{aligned}
C(\sigma^2)&=\int_{\sigma^2}^{\infty} \frac{N}{Mz}-\frac{1}{M}\Tr(z\BI_N+\BH\BH^{H})^{-1}\mathrm{d}z.
\end{aligned}
\end{equation}
When $N\le M$, by (\ref{resol_rate}) in Lemma~\ref{trace_lemma}, we can obtain
\begin{align*}
&\E[C(\rho)]=\int_{\sigma^2}^{\infty}\frac{N}{Mz} - \frac{1}{M}\E[\Tr(z\BI_N+\BH\BH^{H})^{-1}]\mathrm{d}z
\\
&=\int_{\sigma^2}^{\infty} \frac{N}{Mz}-\frac{1}{M}\E[\Tr(\BG(z)(\BX\BX^{H}+\BY\BY^{H}))] \mathrm{d}z \numberthis \label{EC_integ}
\\
&=\int_{\sigma^2}^{\infty} \frac{N}{Mz}-\frac{N}{M}\underbrace{\left(\frac{M\delta_z}{N(1+(1+z)\delta_{z})}+\frac{(n-M)\delta_z}{N(1+z\delta_{z})} \right)}_{K(z)}\mathrm{d}z,
\end{align*}
where $\BG(z)=((1+z)\BX\BX^{H}+z\BY\BY^{H})^{-1}$, $\delta_z$ is given in~(\ref{deltaz_exp}) at the top of next page, and $N_0=n-M$.
\begin{figure*}
\vspace{-0.6cm}
\begin{equation}
\label{deltaz_exp}
\begin{aligned}
\delta_z&=\frac{-(( z+(1+z)   -\frac{N_0z}{N} -(1+z)\frac{M}{N}))-\sqrt{( z+(1+z)   -\frac{N_0z}{N} -(1+z)\frac{M}{N})^2-4(1-\frac{n}{N})z(1+z)}}{2(1-\frac{n}{N})z(1+z)}
\\
&=\frac{-[ \frac{N-M}{n}+(\frac{2N}{n}-1)z ]-\sqrt{(z+\lambda_{-})(z+\lambda_{+})}}{2(\frac{N}{n}-1)z(1+z)}
=\frac{\sqrt{(z+\lambda_{-})(z+\lambda_{+})}+[ \frac{N-M}{n}+(\frac{2N}{n}-1)z ]}{2(1-\frac{N}{n})z(1+z)}.
\end{aligned}
\end{equation}
\vspace{-0.6cm}
\hrulefill
\end{figure*}
It is worth noticing that the interval $[\lambda_{-},\lambda_{+}]$ is exactly the support of the limiting eigenvalue distribution of Jacobi ensemble, which agrees with~\cite[Eq. (21)]{simon2006crossover}. Define $F(z)=\frac{1}{N}[M\log(1+(1+z)\delta_z)+(n-M)\log(1+z\delta_z)-N\log(\delta_z)]$ and notice that 
\begin{align*}
F'(z)&=K(z)+\frac{M(1+z)\delta_z'}{N(1+(1+z)\delta_{z})}+\frac{(n-M)z\delta_z}{N(1+z\delta_{z})}-\frac{\delta_z'}{\delta_z}
\\
&=K(z), \numberthis
\end{align*}
where $K(z)$ is given in~(\ref{EC_integ}). Since $\lim_{z\rightarrow \infty}z\delta=\frac{y_1}{1-y_1}$ and $F(\infty)=0$, we can conclude
\begin{equation}
\begin{aligned}
\E[C(\rho)]= \frac{N}{M}\log(\sigma^{-2})+ \frac{N}{M}F(\sigma^2)=\overline{C}(\sigma^2).
\end{aligned}
\end{equation}
Furthermore, we can obtain $\int_{s=0}^{b}  \frac{\delta}{1+sy_2\delta} \mathrm{d} s=F(b)-F(0)$. It is easy to verify that $F(0)=\frac{(1-y_1)\log(1-y_1)}{y_1}$ so that we can conclude~(\ref{C_exp_1}) in Theorem~\ref{first_the}. The almost sure convergence in~(\ref{as_con}) can be obtained by the convergence of Beta matrices~\cite[Theorem 1.1]{bai2015convergence} since Jacobi matrices belong
to Beta matrices.

\vspace{-0.2cm}

\section{Proof of Theorem~\ref{clt_the}}
\label{prof_the2}
The asymptotic Gaussianity of ID is proved by investigating the convergence of its characteristic function. For ease of presentation, we omit the superscript of $\FS^{(L)}$ and $\BW^{(L)}$ and the subscript of $\BC_L$. The characteristic function of ID $\Psi^{\BW,\BH}_{L}(t)$ is given by
\begin{equation}
\Psi^{\BW,\BH}_{L}(t)=\E[e^{\jmath t \sqrt{LM}I_{N,M,n,L}^{\BW,\BH}(\sigma^2)}]:=\E[\Phi^{\BW,\BH}_{L}(t)].
\end{equation}
Define the notation $N_0=n-M$, $N_1=n-N$ and 
\begin{equation}
\begin{aligned}
\label{BA_proof}
\BQ(z) \!=\!\left(z\BI_N+\BH\BH^{H} \right)^{-1},\widetilde{\BQ}(z) \!=\!\left(z\BI_M+\BH^{H}\BH \right)^{-1},
\end{aligned}
\end{equation}
$\BQ=\BQ(\sigma^2)$, $\widetilde{\BQ}=\widetilde{\BQ}(\sigma^2)$, and $\BG(a)=\left((1+a)\BX\BX^{H}+a\BY\BY^{H} \right)^{-1}$, where $\BX\in \mathbb{C}^{ m \times p}$ and $\BY\in \mathbb{C}^{m \times (n-p)}$ are i.i.d. Gaussian random matrices with zero mean and $\frac{1}{m}$ variance with $m=\min\{N,M \}$ and $p=\max\{M,N \}$. We further define functions
\begin{equation}
\label{M_fun_def}
\mathcal{E}(x)=\min(1,x^2), ~\mathcal{E}(A,x)=\min(A^{-1},x^2).
\end{equation}
We can conclude that for $x \ge 1$, $t>0$, and $A>0$, we have 
\begin{equation}
\begin{aligned}
& e^{-\frac{A t^2}{2}}\int_{0}^{t} y^{\alpha} e^{\frac{A y^2}{2}}\mathrm{d} y   \le t^{\alpha-1} e^{-\frac{A t^2}{2}}\int_{0}^{t}  y e^{\frac{A y^2}{2}} \mathrm{d} y
\\
&
=t^{\alpha-1} A^{-1}(1-e^{-\frac{A t^2}{2}})=\BO( t^{\alpha-1}\mathcal{E}(A,t) ),
\end{aligned}
\end{equation}
where $\mathcal{E}(A,t)$ is given in~(\ref{M_fun_def}). 

To show that the asymptotic distribution of $\sqrt{ML}I_{N,M,n,L}^{\BW,\BH}(\sigma^2)$ converges to the Gaussian distribution, we first show that its characteristic function converges to that of Gaussian distribution, i.e.,
\begin{equation}
\label{cha_con_eq1}
\begin{aligned}
\Psi^{\BW,\BH}_L(t)
=e^{\jmath t \sqrt{ML}\times \overline{C}(\sigma^2) -\frac{t^2 \Xi}{2}  }+J(t,\BC),
\end{aligned}
\end{equation}
where $J(t,\BC)  \xrightarrow[]{{N  \xrightarrow[]{y_1,y_2, \beta}\infty}} 0$ and  $\Xi$ is the asymptotic variance. This approach has been used in the second-order analysis of MIMO channels~\cite{hachem2008new,zhang2022asymptotic,zhang2022secrecy}. Two random matrices, $\BH$ and $\BW$ in the characteristic function $\Psi^{\BW,\BH}_{L}(t)$, will be handled iteratively. The detailed proof is given in the following.

\subsection{Step:1 Compute the expectation with respect to $\BW$}

In this step, we will show that $\frac{\partial \Psi_L^{ \BW,\BH}(t)}{\partial t}$ can be approximated by  $\E [(\jmath \chi^{\BH}_{L}-t\psi^{\BH}_{L} )\Phi^{\BH}_{L}(t)]$ by removing the dependence on $\BW$. As shown in Lemma~\ref{trace_lemma}, the approximation error for $\E[\Tr(\BG(z))]$ over Jacobi channels has the same order as that of Rayleigh~\cite{hachem2008new} and Rayleigh-product MIMO channels~\cite{zhang2022asymptotic} (all are $\BO(N^{-1})$). Thus, the derivative of the characteristic function $ \Psi^{ \BW,\BH}(u)$ can be approximated by the same approach as~\cite[Appendix D.D, Eq. (221) to (240)]{hoydis2015second} and~\cite[Eq. (88)]{zhang2022second} with
\begin{equation}
\label{chara_lemma}
\begin{aligned}
\frac{\partial  \E[\Phi^{\BW,\BH}_{L}(t)]}{\partial t} \!=\!\E [(\jmath \chi^{\BH}_{L}-t\psi^{\BH}_{L} )\Phi^{\BH}_{L}(t)]
+\BO(d(t,1)),
\end{aligned}
\end{equation}
where $d(t,A)=\frac{t^2}{M}+\frac{t^3\mathcal{E}(A,t)}{M^2}+\frac{t^4 \mathcal{E}(A,t)}{M^3}$,
\begin{equation}
\Phi_{L}^{\BH}=e^{ \Upsilon^{\BH}_{L}},
\end{equation}
\begin{equation}
\Upsilon^{\BH}_{L}=\jmath t \chi^{\BH}_L-\frac{t^2 \psi^{\BH}_L}{2}
+\frac{\jmath t^3 \gamma_L^{\BH}}{3},
\end{equation}
\vspace{-0.5cm}
\begin{align*}
\chi^{\BH}_L&=\sqrt{\frac{L}{M}}\log\det\left(\bold{I}_N+\BH\BH^{H}\right)-\frac{L}{\sqrt{ML}}\Tr(\BQ\BH\BC\BH^{H}),
\\
\psi^{\BH}_L&=\frac{L}{ML}\Tr((\BQ\BH\BH^{H})^2)+\frac{2\sigma^2L}{ML}\Tr\Bigl(\BQ^2\BH\frac{\FS\FS^{H}}{L}\BH^{H}\Bigr),
\\
\gamma_L^{\BH}&=\frac{L}{\sqrt{L^3M^3}}\Tr((\BQ\BH\BH^{H})^3)
\\
&
+\frac{3\sigma^2L}{\sqrt{L^3M^3}}\Tr\Bigl(\BQ^2\BH\BH^{H}\BQ\BH\frac{\FS\FS^{H}}{L}\BH^{H}\Bigr),\numberthis
\end{align*}
and $\mathcal{E}(x)$ is defined in~(\ref{M_fun_def}). By taking integral over both sides of~(\ref{chara_lemma}), $ \Psi^{\BW,\BH}_{L}(t)$ can be evaluated by
\begin{equation}
 \Psi^{\BW,\BH}_{L}(t)=\E[ \Phi_{L}^{\BH}(t)] +\BO\Bigl(\frac{t  \mathcal{E}(t)}{M^2}\Bigr).
\end{equation}
Thus, by~(\ref{chara_lemma}), the evaluation for the characteristic function resorts to that for $\E [(\jmath \chi^{\BH}_{L}-t\psi^{\BH}_{L} )\Phi^{\BH}_{L}(t)]$ by taking the expectation with respect to $\BH$. The aim of collecting $t$-terms is to analyze the approximation error of Gaussian approximation.


\subsection{Step 2: Evaluation of $\E  [(\jmath \chi^{\BH}_L-t\psi^{\BH}_L ) \Phi_{L}^{\BH}(t)]$  }
\label{appro_charaf}
In this step, we further remove the dependence of $\E  [(\jmath \chi^{\BH}_L-t\psi^{\BH}_L ) \Phi_{L}^{\BH}(t)]$ on $\BH$ by utilizing the integration by parts formula~\cite[Eq. (17)]{hachem2008new} and the variance control. To this end, we first decompose $\E [ (\jmath\chi_L^{\BH}-t\psi^{\BH}_{L} ) \Phi_{L}^{\BH}]$ into $P_H$ and $P_A$ as
 \begin{equation}
\label{exp_U1U2}
\begin{aligned}
& \E [(\jmath \chi^{\BH}_L-t\psi^{\BH}_L )  \Phi_{L}^{\BH}(t)] =\E[\jmath \chi^{\BH}_L\Phi_{L}^{\BH}(t)]
\\
&
-t\E [\psi^{\BH}_L\Phi_{L}^{\BH}(t)]
=P_H+P_A.
\end{aligned}
\end{equation}
Now we turn to evaluate $P_H$ and $P_A$. 
\subsubsection{Evaluation of $P_H$}
With~\cite[Eq. (4)]{zhang2021bias}, $P_H$ can be represented as follows
\begin{align*}
 P_H & =\jmath\sqrt{\frac{L}{M}}\int_{\sigma^2}^{\infty} \frac{N\E[\Phi_{L}^{\BH}(t)]}{z}
-\E[\Tr(\BQ(z))\Phi_{L}^{\BH}(t)]\mathrm{d}z
\\
&
-\frac{\jmath L}{\sqrt{ML}}\E [\Tr(\BQ\BH\BC\BH^{H})\Phi_{L}^{\BH}(t)]
=P_{H,1}+P_{H,2}. \numberthis \label{PH_step1}
\end{align*}
The proof can be divided into two cases: $N\le M$ and $N>M$. The evaluation of $P_{H,1}$ is same for both cases but those of $P_{H,2}$ are different. 

When $N\le M$, the channel can be equivalently represented by $\BH=(\BX\BX^{H}+\BY\BY^{H})^{-\frac{1}{2}}\BX$. By the integration by parts formula, we have the evaluation for $\E[\Tr(\BQ(z))\Phi_{L}^{\BH}(t)]$ as
\begin{align*}
& \E[\Tr(\BQ(z))\Phi_{L}^{\BH}(t)]=\E[\Tr(\BG(z)(\BX\BX^{H}+\BY\BY^{H}))\Phi_{L}^{\BH}(t)] 
\\
&
=Q_1+Q_2+\BO\Bigl(f(t,\BC)+\frac{t\mathcal{P}(z^{-1})}{z^2N}\Bigr),\numberthis \label{Q_phi_final}
\end{align*}
where  $\BG(z)=\left((1+z)\BX\BX^{H}+z\BY\BY^{H} \right)^{-1}$,
\begin{align*}
Q_1 &= (M\overline{C}'(z)+N z^{-1} )\E[\Phi^{\BH}_{L}(t)],\numberthis
\\
Q_2 &=\jmath  t   \sqrt{\frac{L}{M}} \Bigl[-\log\Bigl(\frac{\frac{ M(1+\sigma^2)}{N}}{(1+(1+\sigma^2)\delta_{\sigma^2})(1+(1+z)\delta_z)}
\\
&
+ \frac{\frac{ N_0 \sigma^2}{N}}{(1+\sigma^2\delta_{\sigma^2})(1+z\delta_z)}\Bigr)
+\log\Bigl(\frac{\frac{ M}{N}}{(1+\delta_{0})(1+(1+z)\delta_z)}
\\
&
+ \frac{\frac{ N_0 }{N}}{(1+\delta_{0})(1+z\delta_z)}\Bigr)\Bigr]':=\jmath  t   \sqrt{\frac{L}{M}} V'(z),\numberthis
\end{align*}
$\delta_z=\delta(1+z,z)$, and $\mathcal{P}(\cdot)$ is a polynomial with positive coefficients. Thus, $P_{H,1}$ can be obtained by taking integral over $z$ as
\begin{equation}
\label{PH1_eva}
\begin{aligned}
P_{H,1}&= 
\Bigl[-\jmath \sqrt{ML} \overline{C}(z)  |_{\sigma^2}^{\infty}+\frac{Lt}{M} V(z)|_{\sigma^2}^{\infty})\Bigr] \E[\Phi^{\BH}_L(t)]
\\
&
+\BO\Bigl(\frac{\mathcal{P}(\sigma^{-2})}{\sigma^2}\Bigl(\frac{t}{M}+f(t,\BC)\Bigr)\Bigr)
\\
&=\Bigl(\jmath \sqrt{ML} \overline{C}(\sigma^2)-\frac{tL}{M}V_1\Bigr)\E[\Phi^{\BH}_L(t)]
\\
&
+\BO\Bigl(\frac{\mathcal{P}(\sigma^{-2})}{\sigma^2}\Bigl(\frac{t}{M}+f(t,\BC)\Bigr)\Bigr),
\end{aligned}
\vspace{-0.1cm}
\end{equation}
where $V_1$ is given in~(\ref{var_exp1}) and
\begin{equation}
\begin{aligned}
f(t,\BC)&=\frac{1}{N}+\frac{t\sqrt{\Tr(\BC^2)}}{N^{\frac{3}{2}}}+\frac{t^2\Bigl(N+\sqrt{\frac{\Tr(\BC^2)}{N}} \Bigr)}{N^2}
\\
&+\frac{t^3 \Bigl(N+\sqrt{\frac{\Tr(\BC^2)}{N}} \Bigr)}{N^3}.
\end{aligned}
\end{equation}
When $N\le M$, by the integration by parts formula and variance control, $P_{H,2}$ can be evaluated by
\begin{align*}
&P_{H,2}\!=\!-\frac{  t L\E[\Tr(\BG^2(\sigma^2))]}{MN (1+(1+\sigma^2)\delta_{\sigma^2})^4}\frac{\Tr(\BC^2)}{N} \E[\Phi^{\BH}_L(t)]\!+\!\BO(g(t,\BC))
\\
&\!=\!-\frac{ t L}{M (1+(1+\sigma^2)\delta_{\sigma^2})^4} \frac{M\delta_{\sigma^2}\frac{\Tr(\BC^2)}{M} \E[\Phi^{\BH}_L(t)]}{N \Bigl(\frac{\frac{(1+\sigma^2) M}{N}}{(1+(1+\sigma^2)\delta_{\sigma^2})^2}+ \frac{\frac{\sigma^2  N_0}{N}}{(1+\sigma^2\delta_{\sigma^2})^2} \Bigr) }
\\
&
+\BO(g(t,\BC))  \numberthis \label{PH2_eva}
\\
&
\!=\! -\frac{ t \beta V_3\Tr(\BC^2)}{M} \E[\Phi^{\BH}_L(t)]+\BO(g(t,\BC)),
\end{align*}
where
\vspace{-0.2cm}
\begin{equation}
\begin{aligned}
g(t,\BC)&=\frac{\sqrt{\Tr(\BC^2)}}{N^{\frac{3}{2}}}
+  \frac{t\Tr(\BC^2)}{N^{\frac{3}{2}}} 
\\
&+ \frac{t^2(1+ \frac{\Tr(\BC^2)}{N} )}{N}+\frac{t^3(1+ \frac{\Tr(\BC^2)}{N} )}{N^2},
\end{aligned}
\end{equation}
and the last step follows from the evaluation for $\E[\Tr(\BG(\sigma^2))]$ in~(\ref{second_resol}). When $M<N$, we can obtain that $\BH^{H}\BH$ has the same eigenvalue distribution as $\BX\BX^{H}\left(\BX\BX^{H}+\BY\BY^{H}\right)^{-1}$, where $\BX\in\mathbb{C}^{M\times N}$ and $\BY\in\mathbb{C}^{M\times (n-N)}$ are i.i.d. complex Gaussian random matrices. In this case, $\BG(\sigma^2)\in \mathbb{C}^{M\times M}$ and $P_{H,2}$ can be evaluated by
\begin{align*}
&P_{H,2}=-  \beta \E [\Tr(\widetilde{\BQ}\BH^{H}\BH\BC)\Phi_{L}^{\BH}(t)]
\\
&
=-  t\beta \E[\Tr(\BG(\sigma^2)\BX\BX^{H}\BC)\Phi^{\BH}_{L}(t)]
\\
&
=-  t\beta\frac{N N_1\sigma^2\delta_{\sigma^2}^2}{M^2(1+(1+\sigma^2)\delta_{\sigma^2})^2(1+\sigma^2\delta_{\sigma^2})} \numberthis \label{QXXC_eva}
\\
&
\times
\Bigl[\delta-\frac{N}{M(1+\sigma^2\delta_{\sigma^2})(1+(1+\sigma^2)\delta_{\sigma^2})^2}\frac{\E[\Tr(\BG^2(\sigma))]}{M} \Bigr]
\\
&\times 
 \E[\Phi^{\BH}_{L}(t)]+g(t,\BC)
\\
&
=- \frac{  t\beta V_3 \Tr(\BC^2)}{M} \E[\Phi^{\BH}_{L}(t)]+g(t,\BC).
\end{align*}
\subsubsection{Evaluation of $P_{A}$}
When $N \le M$, by the resolvent identity $\BI_N=\sigma^2\BQ+\BQ\BH\BH^{H}$, we can rewrite $\E[\psi^{\BH}_L\Phi_{L}^{\BH}(t)]$ as
\begin{align*}
& \E[\psi^{\BH}_L\Phi_{L}^{\BH}(t)]=\frac{N}{M}\E\Bigl[(\frac{1}{N}\Tr(\BI_N-\sigma^4\BQ^2 -2\BQ^2\BH\BH^{H}  ) 
\\
&
+\frac{2\sigma^2}{N}\Tr(\BQ^2\BH\frac{\FS\FS^{H}}{L}\BH^{H}))\Phi_{L}^{\BH}(t)\Bigr]+\BO\Bigl(\frac{1}{M}\Bigr)
\\
&=\frac{N}{M}\E\Bigl[\!1 \!\!-\!\frac{\sigma^4\Tr(\BQ^2)}{N}\! \!+\!\!\frac{2\sigma^2}{N}\Tr(\BQ^2\BH\BC\BH^{H}))\Phi_{L}^{\BH}(t)\! \Bigr]\!\!+\!\BO\Bigl(\frac{1}{M}\Bigr)
\\
&\overset{(a)}{=} \frac{N}{M} \Bigl(1-\frac{\sigma^4\E[\Tr(\BQ^2)]}{N}\Bigr)\E[\Phi_{L}^{\BH}(t)]+\BO\Bigl(\frac{\Tr(\BC^2)}{N^{\frac{3}{2}}}\Bigr)
\\
&=\frac{N}{M} \Bigl(1+\frac{n-N}{N}\sigma^4\delta'_{\sigma^2}  \Bigr)\E[\Phi_{L}^{\BH}(t)]+\BO\Bigl(\frac{\Tr(\BC^2)}{N^{\frac{3}{2}}}\Bigr)
\\
&=V_2\E[\Phi_{L}^{\BH}(t)]+\BO\Bigl(\frac{\Tr(\BC^2)}{N^{\frac{3}{2}}}\Bigr),\numberthis \label{Ppsi_eva}
\end{align*}
where the approximation error in step $(a)$ follows from the variance control similar to~\cite[Eq. (138)]{zhang2022second}. The case for $M<N$ can be handled similarly by evaluating $\E[\Tr(\widetilde{\BQ})]$. By~(\ref{chara_lemma}),~(\ref{PH1_eva}),~(\ref{PH2_eva}), and~(\ref{Ppsi_eva}), we can obtain the following differential equation
\begin{equation}
\label{chara_lemma2}
\frac{\partial  \E[\Phi^{\BW,\BH}_{L}\! (t)]}{\partial t} \!\!=\!(\jmath \sqrt{ML}\overline{C}(\sigma^2)\!-\! t \Xi )\E [\Phi^{\BW,\BH}_{L}(t)]+o(1).
\end{equation}

\subsection{Step 3: Convergence of $\E \bigl[ e^{\jmath\frac{\sqrt{ML}}{\sqrt{\Xi}}(I_{N,M,n,L}^{\BW,\BH}(\sigma^2)-\overline{C}(\sigma^2))}\bigr]$}
By solving the differential equation in~(\ref{chara_lemma2}), we can obtain the evaluation for $\Psi^{\BW,\BH}(t)$ as
\begin{align*}
&\Psi^{\BW,\BH}(t)=e^{\jmath t \sqrt{ML}\times \overline{C}(\sigma^2) -\frac{t^2 \Xi}{2}  }
\\
&
\times
(1+\int_{0}^{t}e^{-\jmath s \sqrt{ML}\times \overline{C}(\sigma^2) +\frac{s^2 \Xi}{2}  }\overline{\varepsilon}(s,\BC)\mathrm{d}s)\numberthis \label{pphi_eps}
\\
&
=e^{\jmath t \sqrt{ML}\times \overline{C}(\sigma^2) -\frac{t^2 \Xi}{2}  }+J(t,\BC),
\end{align*}
where $\Xi=\beta V_1+V_2+\frac{\Tr(\BC^2)}{M}\beta V_3$, $\overline{\varepsilon}(t,\BC)=\BO\bigl(\frac{t\mathcal{E}(\Xi,t)}{N}+\frac{t\Xi\mathcal{E}(\Xi,t)}{N^2}+d(t,\Xi)+\frac{t}{M}+f(t,\BC)+g(t,\BC)+\frac{\Tr(\BC^2)}{N^{\frac{3}{2}}}\bigr)$, and the error term $J(t,\BC)$ is given by
\begin{align*}
& J(t,\BC)=\BO\Bigl(\frac{\mathcal{E}(\Xi,t)}{N}+\frac{ \Xi\mathcal{E}(\Xi,t)}{N^2}
\\
&+ \frac{t\mathcal{E}(\Xi,t)}{N}+\frac{t^2 \mathcal{E}(\Xi,t)}{N^2}+\frac{t^3\mathcal{E}(\Xi,t)}{N^3}
\\
&+
\frac{t}{N} +\frac{t\sqrt{\Tr(\BC^2)}}{N^\frac{3}{2}}
+\frac{t \mathcal{E}(\Xi,t) (N+\sqrt{\frac{\Tr(\BC^2)}{M}})}{N^{2}}   \numberthis \label{final_error_term}
\\
&
+\frac{t^2 \mathcal{E}(\Xi,t)(N+\sqrt{\frac{\Tr(\BC^2)}{M}})}{N^{3}}+
\frac{\mathcal{E}(\Xi,t)}{ N}+\frac{t\sqrt{\Tr(\BC^2)}}{N^{\frac{3}{2}}} 
\\
&
+\Bigl[\frac{\Tr(\BC^2)}{N^{\frac{3}{2}}} + \frac{t(1+\frac{\Tr(\BC^2)}{M})}{N}+\frac{t^2(1+\frac{\Tr(\BC^2)}{M})}{N^2}\Bigr]\mathcal{E}(\Xi,t)\Bigr).
\end{align*}
Therefore, the characteristic function of the normalized ID can be written as
\begin{align*}
&\Psi^{\BW,\BH}_{norm}(t) =\E [e^{{\jmath\frac{\sqrt{ML}}{\sqrt{\Xi}}(I_{N,M,n,L}^{\BW,\BH}(\sigma^2)-\overline{C}(\sigma^2))}}]\numberthis \label{conver_cha}

\\
&
=\Psi^{\BW,\BH}\Bigl(\frac{t}{\sqrt{\Xi}}\Bigl)e^{-\jmath t\frac{\overline{C}(\sigma^2)}{\sqrt{\Xi}}}
+J\Bigl(\frac{t}{\sqrt{\Xi}},\BC\Bigr)
\overset{(a)}{=}e^{-\frac{t^2}{2}}\!\!+\!\BO\Bigl(\frac{1}{\sqrt{N}}\Bigr),
\end{align*}
where step $(a)$ in~(\ref{conver_cha}) follows from
\begin{equation}
\begin{aligned}
& \BO\Bigl(\mathcal{E}(\frac{t}{\sqrt{\Xi}},\BC)\Bigr) =\BO\Bigl(\frac{\mathcal{E}(t)\frac{\Tr(\BC^2)}{M}}{\sqrt{M}\Xi}\Bigr)
\\
&
=\BO\Biggl(\frac{\mathcal{E}(t)\frac{\Tr(\BC^2)}{M}}{\sqrt{M}(\frac{\tau\Tr(\BC^2)}{M}+K)}\Biggr)=\BO\Bigl(\frac{\mathcal{E}(t)}{\sqrt{N}}\Bigr).
\end{aligned}
\end{equation} 
Here $\tau=\BO(1)$ is the coefficient of $\frac{\Tr(\BC^2)}{M}$ and $K$ is a constant independent of $M$, $N$, $L$, and $n$. By~(\ref{conver_cha}) and L{\'e}vy’s continuity theorem~\cite{billingsley2017probability}, we can obtain the following convergence
\begin{equation}
\begin{aligned}
\sqrt{\frac{ML}{\Xi}}(I_{N,M,n,L}^{\BW,\BH}(\sigma^2)\!-\!\overline{C}(\sigma^2)) \!\!\xrightarrow[{N  \xrightarrow[]{(y_1,y_2, \beta)}\infty}]{\mathcal{D}}\!\! \mathcal{N}(0,1).
\end{aligned}
\end{equation}

\subsection{Approximation accuracy of~(\ref{prob_con_rate})}

The convergence rate for the Gaussian approximation of ID can be obtained by the Esseen inequality~\cite[p538]{feller1991introduction}, which says that
 there exists $C>0$ for any $T>0$ such that 
\begin{align*}
&\sup\limits_{x\in \mathbb{R}}\Biggl|\Prob(\sqrt{\frac{{ML}}{\Xi}}(I_{N,M,n,L}^{\BW,\BH}(\sigma^2)-\overline{C}(\sigma^2)) \le x  )\Biggr|
\\
&
\le C \int_{0}^{T} t^{-1} | \Psi^{\BW,\BH}_{norm}(t) - e^{-\frac{t^2}{2}}| \mathrm{d} t+ \frac{C}{T} 
\\
& \le K \Bigl(\int_{0}^{T} t^{-1}  J(\frac{t}{\sqrt{\Xi}},\BC) +\frac{1}{T}\Bigl). \numberthis  \label{berry_ineq}
\end{align*}

Notice that the dominating term in~(\ref{final_error_term}) is $\BO(\frac{\mathcal{E}(t)+t}{\sqrt{M}})$ and for $T>1$,
\begin{equation}
\begin{aligned}
& \int_{0}^{T}t^{-1} (\mathcal{E}(1,t)+t) \mathrm{d} t
=\int_{0}^{1}(1+t) \mathrm{d} t
\\
&
+\int_{1}^{T}t^{-1} (1+t) \mathrm{d} t = \BO(\log(T)+T)=\BO(T).
\end{aligned}
\end{equation}
By taking $T=L^{\frac{1}{4}}$ in~(\ref{berry_ineq}), we can obtain
 \begin{equation}
\sup\limits_{x\in \mathbb{R}}\Bigl|\Prob\Bigl(\sqrt{\frac{{ML}}{\Xi}}(I_{N,M,n,L}^{\BW,\BH}(\sigma^2)-\overline{C}(\sigma^2)) \le x  \Bigr)\Bigr|=\BO(L^{-\frac{1}{4}}),
\end{equation}
which concludes the proof of Theorem~\ref{clt_the}.

\ifCLASSOPTIONcaptionsoff
  \newpage
\fi



\bibliographystyle{IEEEtran}
\bibliography{IEEEabrv,ref}
\end{document}